\documentclass[11pt]{article}
\usepackage[margin=1in]{geometry}

\usepackage{adjustbox}
\usepackage{amsmath}
\usepackage{amssymb}
\usepackage{amsthm}
\usepackage[title]{appendix}
\usepackage{authblk}
\usepackage{bbm}
\usepackage{booktabs}
\usepackage[justification=centering]{caption}
\usepackage{derivative}
\usepackage[shortlabels]{enumitem}
\usepackage{graphicx}
\usepackage[colorlinks=true,linkcolor=blue,citecolor=blue]{hyperref}
\usepackage{makecell}
\usepackage{mathtools}
\usepackage{multirow}
\usepackage{natbib}
\usepackage{xcolor}

\newtheorem{theorem}{Theorem}

\newtheorem{lemma}{Lemma}

\theoremstyle{remark}
\newtheorem{remark}{Remark}
\theoremstyle{definition}

\newtheorem{assumption}{Assumption}
\newtheorem{example}{Example}

\makeatletter
\newcommand{\neutralize}[1]{\expandafter\let\csname c@#1\endcsname\count@}
\makeatother

\newenvironment{assumptionp}[1]
  {\renewcommand{\theassumption}{\ref*{#1}$'$}%
   \neutralize{assumption}\phantomsection
   \begin{assumption}}
  {\end{assumption}}

\def\E{\mathbb{E}}
\def\R{\mathbb{R}}
\def\Dsc{\mathcal{D}}
\def\Hsc{\mathcal{H}}
\def\Msc{\mathcal{M}}
\def\Nsc{\mathcal{N}}
\def\Usc{\mathcal{U}}

\def\mhat{\widehat{m}}
\def\Deltahat{\widehat{\Delta}}
\def\etahat{\widehat{\eta}}
\def\sigmahat{\widehat{\sigma}}
\def\xihat{\widehat{\xi}}
\def\varepsilonhat{\widehat{\varepsilon}}

\def\convd{\overset{d}{\longrightarrow}}
\def\convp{\overset{p}{\longrightarrow}}
\DeclareMathOperator{\var}{var}
\DeclareMathOperator{\cov}{cov}
\DeclareMathOperator{\tr}{tr}

\providecommand{\keywords}[1]{{\small\textit{Keywords:} #1}}

\begin{document}

\title{Machine-Learning-Assisted Comparison of Regression Functions}
\author[1]{Jian Yan}
\author[2]{Zhuoxi Li}
\author[1]{Yang Ning}
\author[3]{Yong Chen}
\date{}
\affil[1]{Department of Statistics and Data Science, Cornell University}
\affil[2]{Department of ISOM, Hong Kong University of Science and Technology}
\affil[3]{Department of Biostatistics, Epidemiology and Informatics, University of Pennsylvania}

\maketitle

\begin{abstract}
We revisit the classical problem of comparing regression functions, a fundamental question in statistical inference with broad relevance to modern applications such as data integration, transfer learning, and causal inference. Existing approaches typically rely on smoothing techniques and are thus hindered by the curse of dimensionality. We propose a generalized notion of kernel-based conditional mean dependence that provides a new characterization of the null hypothesis of equal regression functions. Building on this reformulation, we develop two novel tests that leverage modern machine learning methods for flexible estimation. We establish the asymptotic properties of the test statistics, which hold under both fixed- and high-dimensional regimes. Unlike existing methods that often require restrictive distributional assumptions, our framework only imposes mild moment conditions. The efficacy of the proposed tests is demonstrated through extensive numerical studies. 
\end{abstract}
\keywords{Comparison of regression functions; High dimensionality; Kernel-based conditional mean dependence; Machine learning.}

\section{Introduction}
The comparison of two regression functions is a fundamental problem in regression analysis \citep{munk1998nonparametric}. Consider two independent data sets $\Dsc^{(l)}=\{(Y_{i}^{(l)},X_{i}^{(l)})\}_{i=1}^{N_{l}}$ for $l=1,2$, where $Y_{i}^{(l)}\in\R$ is the response and $X_{i}^{(l)}\in\R^{p}$ is a set of covariates with dimension $p$ allowed to diverge. For $l=1,2$, we assume that $\{(Y_{i}^{(l)},X_{i}^{(l)})\}_{i=1}^{N_{l}}$ are independent and identically distributed samples of $(Y^{(l)},X^{(l)})\sim P^{(l)}\equiv P_{Y\mid X}^{(l)}\otimes P_{X}^{(l)}$, and define the regression function $m^{(l)}(x)=\E(Y^{(l)}\mid X^{(l)}=x)$ and the error $\varepsilon^{(l)}=Y^{(l)}-m^{(l)}(X^{(l)})$. We aim to test
\begin{equation}\label{eq:null}
    H_{0}:P^{(1)}_{X}\left\{m^{(1)}(X)=m^{(2)}(X)\right\}=1\quad\text{versus}\quad H_{a}:P^{(1)}_{X}\left\{m^{(1)}(X)\neq m^{(2)}(X)\right\}>0. 
\end{equation}
To ensure that the testing problem is nontrivial, we assume that $P_{X}^{(1)}\ll P_{X}^{(2)}$ and $P_{X}^{(2)}\ll P_{X}^{(1)}$, where the symbol $\ll$ stands for absolute continuity. Thus, the hypotheses (\ref{eq:null}) can be equivalently stated by replacing $P^{(1)}_{X}$ with $P^{(2)}_{X}$. 

The problem (\ref{eq:null}) also naturally arises in a variety of related areas. For instance, in data integration, it is necessary to assess whether two data sets share a common regression function before proceeding with modeling, estimation, and inference. When the null hypothesis $H_{0}$ holds, pooling information across sources can enhance statistical efficiency and lead to more reliable conclusions. In contrast, when heterogeneity is present, naive aggregation of data sets may obscure meaningful differences and result in misleading scientific findings. A similar issue appears in transfer learning, where it is commonly assumed that the conditional mean is identical in the source and target populations \citep{wang2025phase}—a weaker version of the standard covariate shift assumption. Testing hypotheses (\ref{eq:null}) is thus crucial for validating methods built upon this assumption. Another example comes from causal inference. In the binary treatment setting, testing the null hypothesis of a zero conditional average treatment effect \citep{crump2008nonparametric} can be formulated as problem (\ref{eq:null}) under the standard assumptions of consistency and no unmeasured confounding. In this context, the two independent samples are naturally defined by the treatment indicator. 

Much effort has been devoted to the problem (\ref{eq:null}) in the literature; see Section 7 of \citet{gonzalez2013updated} for a comprehensive review. Classical methods employ parametric models for the regression functions and assess equality through comparisons of model parameters. However, this approach necessitates correct model specification, which is often unrealistic in practice. The problem of testing the equality of two regression functions in the nonparametric setting was first considered by \citet{hall1990bootstrap,king1991testing}, and subsequently explored in a series of studies including \citet{delgado1993testing,kulasekera1995comparison,young1995non,dette2001nonparametric,lavergne2001equality,neumeyer2003nonparametric,pardo2007testing,srihera2010nonparametric,pardo2015non}. The existing literature primarily focuses on the univariate case, i.e., $p=1$. Most available methods rely on smoothing-based estimators of the regression functions, and thus suffer from the curse of dimensionality even if they can be extended to multivariate covariates. Furthermore, restrictive distributional assumptions are often imposed to facilitate theoretical analysis. \citet{hall1990bootstrap,king1991testing,delgado1993testing} concentrate on identical design points. \citet{king1991testing,young1995non} assume Gaussian errors. \citet{kulasekera1995comparison} assumes independence between $\varepsilon^{(l)}$ and $X^{(l)}$, and \citet{dette2001nonparametric,neumeyer2003nonparametric,pardo2007testing,pardo2015non} posit a heteroscedastic model of the form $\varepsilon^{(l)}=\sigma^{(l)}(X^{(l)})e^{(l)}$, where $e^{(l)}$ is independent of $X^{(l)}$ and $\sigma^{(l)}(\cdot)$ denotes the variance function. As noted by \citet{racine2020smooth}, these independence assumptions may be unduly stringent and should themselves be subjected to empirical testing. Moreover, several studies \citep{dette2001nonparametric,neumeyer2003nonparametric,pardo2007testing,pardo2015non} assume that $X^{(l)}$ has compact support (e.g., $[0,1]$) with density bounded away from zero, thereby excluding frequently encountered Gaussian distributions. 

In contrast to traditional smoothing-based methods, modern machine learning techniques offer flexible and powerful tools for modeling regression functions while alleviating the curse of dimensionality. In this paper, we propose two tests that leverage machine learning methods to address problem (\ref{eq:null}). To this end, we introduce a generalized notion of kernel-based conditional mean dependence which yields a novel characterization of the null hypothesis. Our proposed tests possess several attractive features relative to existing approaches:
\begin{itemize}
    \item They remain valid irrespective of whether the covariate dimension $p$ is fixed or diverging. 

    \item They do not rely on any specific functional form of the regression functions and require only mild moment conditions on the covariates and errors.

    \item They allow for heterogeneous covariate and error distributions across the two samples. 

    \item Under $H_{0}$, the test statistics admit simple Gaussian limiting distributions, allowing straightforward implementation without resorting to resampling procedures. Additionally, users retain full flexibility in the choice of regression methods. 
\end{itemize}

\textit{Notation.} Denote by $z_{\alpha}$ the $\alpha$ quantile of the standard normal distribution, for $\alpha\in(0,1)$. For two sequences of real numbers $\{a_{n}\}_{n=1}^{\infty}$ and $\{b_{n}\}_{n=1}^{\infty}$, we write $a_{n}=O(b_{n})$ if there exists a finite $n_{0}\ge 1$ and a finite $C>0$ such that $|a_{n}|\le C|b_{n}|$ for all $n\ge n_{0}$; $a_{n}=o(b_{n})$ if $\lim_{n\rightarrow\infty}a_{n}/b_{n}=0$; $a_{n}=\Omega(b_{n})$ if $b_{n}=O(a_{n})$; $a_{n}=\omega(b_{n})$ if $b_{n}=o(a_{n})$; $a_{n}\asymp b_{n}$ if $a_{n}=O(b_{n})$ and $b_{n}=O(a_{n})$. The symbols $\overset{p}{\rightarrow}$ and $\overset{d}{\rightarrow}$ stand for convergence in probability and in distribution, respectively. Denote by $\|\cdot\|$ the Euclidean norm in $\R^{p}$. For any probability measure $P$ on $\R^{p}$, let $L^{2}(P)$ denote the Hilbert space of all square integrable functions on $\R^{p}$. For $f\in L^{2}(P)$, denote by $\|f\|_{L^{2}(P)}=(\int f^{2}dP)^{1/2}$ the $L^{2}(P)$ norm of $f$.

\section{Preliminaries}
This section reviews reproducing kernel Hilbert spaces and introduces a generalized notion of the kernel-based conditional mean dependence. 

Let $\Hsc$ be a Hilbert space of real-valued functions defined on a topological space $\Usc$. A function $k:\Usc\times\Usc\rightarrow\R$ is called a reproducing kernel of $\Hsc$ if (i) $\forall u\in\Usc$, $k(\cdot,u)\in\Hsc$, and (ii) $\forall u\in\Usc,\forall f\in\Hsc$, $\langle f,k(\cdot,u)\rangle_{\Hsc}=f(u)$, where $\langle\cdot,\cdot\rangle_{\Hsc}$ is the inner product associated with $\Hsc$. If $\Hsc$ has a reproducing kernel, it is said to be a reproducing kernel Hilbert space (RKHS). According to the Moore-Aronszajn theorem, for every symmetric, positive definite function (henceforth kernel) $k:\Usc\times\Usc\rightarrow\R$, there is an associated RKHS $\Hsc_{k}$ with reproducing kernel $k$. For $\theta>0$, define $\Msc_{k}^{\theta}(\Usc)=\{\mu\in\Msc(\Usc):\int k^{\theta}(u,u)d|\mu|(u)<\infty\}$, where $\Msc(\Usc)$ denotes the set of all finite signed Borel measures on $\Usc$. The kernel mean embedding of $P\in\Msc_{k}^{1/2}(\Usc)$ into the RKHS $\Hsc_{k}$ is defined by the Bochner integral $\Pi_{k}(P)=\int k(\cdot,u)dP(u)\in\Hsc_{k}$. The kernel $k$ is said to be characteristic if the mapping $\Pi_{k}$ is injective. Conditions under which kernels are characteristic have been studied by \citet{sriperumbudur2008injective,sriperumbudur2010hilbert}. In particular, the Gaussian kernel and the Laplace kernel on $\R^{p}$ are characteristic. 

Let $V\in\R$ be an arbitrary random variable and $U\in\Usc$ be an arbitrary random element. For a given kernel $k:\Usc\times\Usc\rightarrow\R$, under the assumptions that $\E(V^{2})<\infty$ and $U\sim P\in\Msc_{k}^{1}(\Usc)$, the kernel-based conditional mean dependence (KCMD) of $V$ on $U$ is defined as \citep{lai2021kernel}
\[
\text{KCMD}(V\mid U)=\E[\{V-\E(V)\}\{V'-\E(V)\}k(U,U')],
\]
where $(V',U')$ denotes an independent copy of $(V,U)$. The martingale difference divergence proposed in \citet{shao2014martingale} is a special case of the KCMD associated with the kernel induced by the Euclidean distance \citep{lai2021kernel}. When the kernel $k$ is characteristic, we have $\text{KCMD}(V\mid U)\ge 0$, and $\text{KCMD}(V\mid U)=0$ if and only if $\E(V\mid U)=\E(V)$ almost surely. 

The definition of the KCMD can be generalized as follows. For a given $v_{0}\in\R$, define
\[
\text{KCMD}^{*}(V\mid U)=\E\{(V-v_{0})(V'-v_{0})k(U,U')\},
\]
which replaces $\E(V)$ with $v_{0}$. 
\begin{lemma}\label{lemma:kcmd}
    Assume that $\E(V^{2})<\infty$ and $U\sim P\in\Msc_{k}^{1}(\Usc)$. When the kernel $k$ is characteristic, for any $v_{0}\in\R$, we have $\textup{KCMD}^{*}(V\mid U)\ge0$, and $\textup{KCMD}^{*}(V\mid U)=0$ if and only if $\E(V\mid U)=v_{0}$ almost surely. 
\end{lemma}
This generalization of the KCMD, which is not considered in \citet{lai2021kernel}, serves as the foundation of our proposed test statistics. The proof of Lemma \ref{lemma:kcmd} is provided in the appendix.

\section{Test statistic and asymptotic properties}\label{sec:test}
We first make a new reformulation of the null hypothesis $H_{0}$. Define
\begin{align*}
    \eta^{(1)}&=Y^{(1)}-m^{(2)}(X^{(1)}),\\
    \eta^{(2)}&=Y^{(2)}-m^{(1)}(X^{(2)}),
\end{align*}
which become the true errors $\varepsilon^{(1)}$ and $\varepsilon^{(2)}$ under $H_{0}$. Then $H_{0}$ holds if and only if
\begin{equation}\label{eq:equiv}
    \E(\eta^{(l)}\mid X^{(l)})=0\quad(l=1,2)\quad\text{almost surely}.
\end{equation}
Based on this reformulation, for a given kernel $k$ on $\R^{p}$, define
\[
\Delta^{(l)}=\text{KCMD}^{*}(\eta^{(l)}\mid X^{(l)})=\E\left\{\eta^{(l)}\eta^{(l)\prime}k(X^{(l)},X^{(l)\prime})\right\}\quad(l=1,2),
\]
where $(\eta^{(l)\prime},X^{(l)\prime})$ denotes an independent copy of $(\eta^{(l)},X^{(l)})$. In this paper, we make the following assumption on our kernel. Examples include the Gaussian kernel and the Laplace kernel. 
\begin{assumption}\label{assu1}
    The kernel $k$ is characteristic. Besides, there exists a constant $K>0$ such that $\sup_{x,x'}|k(x,x')|\le K$. 
\end{assumption}

\begin{theorem}\label{thm:kcmd}
    Assume that $\E(|\eta^{(l)}|^{2})<\infty$ for $l=1,2$. Under Assumption \ref{assu1}, we have $\Delta^{(l)}\ge 0$, and $\Delta^{(l)}=0$ if and only if the null hypothesis $H_{0}$ in (\ref{eq:null}) holds. 
\end{theorem}

We now introduce our test statistic:
\begin{enumerate}[(i)]
    \item For $l=1,2$, estimate the regression function $m^{(l)}(\cdot)$ using machine learning methods based on the data set $\Dsc^{(l)}$, denoted by $\mhat^{(l)}(\cdot)$. 
    \item Obtain
    \begin{align*}
        \etahat^{(1)}_{i}&=Y^{(1)}_{i}-\mhat^{(2)}(X^{(1)}_{i})\quad(i=1,\ldots,N_{1}),\\
        \etahat^{(2)}_{i}&=Y^{(2)}_{i}-\mhat^{(1)}(X^{(2)}_{i})\quad(i=1,\ldots,N_{2}).
    \end{align*}
    Note that $\etahat^{(l)}_{i}$ is different from the residual $\varepsilonhat^{(l)}_{i}=Y^{(l)}_{i}-\mhat^{(l)}(X^{(l)}_{i})$. 
    \item Without loss of generality, we assume that $N_{l}=2n_{l}\ (l=1,2)$. Define
    \[
    \Deltahat^{(l)}=\frac{1}{n_{l}}\sum_{i=1}^{n_{l}}\etahat^{(l)}_{i}\etahat^{(l)}_{i+n_{l}}k(X^{(l)}_{i},X^{(l)}_{i+n_{l}}).
    \]
    \item Construct the test statistic
    \[
    T=\Deltahat^{(1)}+\Deltahat^{(2)}.
    \]
\end{enumerate}

We make the following assumptions throughout the analysis. 
\begin{assumption}\label{assu2}
    For $l=1,2$, assume that $\E(|\varepsilon^{(l)}|^{4})\le C_{1}<\infty$ for some $C_{1}>0$, and 
    \[
    \E\left\{\left|\varepsilon^{(l)}\varepsilon^{(l)\prime}k(X^{(l)},X^{(l)\prime})\right|^{2}\right\}\ge c_{1}, 
    \]
    for some $c_{1}>0$, where $(\varepsilon^{(l)\prime},X^{(l)\prime})$ denotes an independent copy of $(\varepsilon^{(l)},X^{(l)})$. 
\end{assumption}
\begin{assumption}\label{assu3}
    Suppose
    \begin{gather*}
        \E\left[\left\{\mhat^{(1)}(X^{(2)})-m^{(1)}(X^{(2)})\right\}^{2}\mid\Dsc^{(1)}\right]=o_{p}(n_{1}^{-1/2}),\\
        \E\left[\left\{\mhat^{(2)}(X^{(1)})-m^{(2)}(X^{(1)})\right\}^{2}\mid\Dsc^{(2)}\right]=o_{p}(n_{2}^{-1/2}).
    \end{gather*}
\end{assumption}
We emphasize that the condition $n_{1}\asymp n_{2}$ is not required. Hence, our procedure remains valid under imbalanced sample sizes. Assumption \ref{assu2} guarantees the applicability of the Lyapunov central limit theorem (CLT). By comparison, \citet{zhang2018conditional,li2023testing,he2025goodness} impose a stronger condition: for some constants $c_{2}$ and $C_{2}$, 
\[
0<c_{2}\le\E(|\varepsilon^{(l)}|^{2}\mid X^{(l)})\le\E^{1/2}(|\varepsilon^{(l)}|^{4}\mid X^{(l)})\le C_{2}<\infty\quad\text{almost surely},
\]
which directly implies Assumption \ref{assu2}. Assumption \ref{assu3} accommodates potential covariate shift. In the special case where $P^{(1)}_{X}=P^{(2)}_{X}$, Assumption \ref{assu3} reduces to
\[
\E\left[\left\{\mhat^{(l)}(X^{(l)})-m^{(l)}(X^{(l)})\right\}^{2}\mid\Dsc^{(l)}\right]=o_{p}(n_{l}^{-1/2})\quad(l=1,2),
\]
a condition commonly assumed in the literature on double/debiased machine learning and non- or semiparametric estimation \citep{chernozhukov2018double,kennedy2024semiparametric}. Flexible machine learning methods can be employed to estimate the regression functions $m^{(l)}(\cdot)$. Importantly, Assumption \ref{assu3} is substantially weaker than conditions requiring uniform convergence in the supremum norm, as in \citet{zhang2023classification,he2025goodness}. In fact, a sufficient condition for Assumption \ref{assu3} is
\[
\sup_{x}\left|\mhat^{(l)}(x)-m^{(l)}(x)\right|=o_{p}(n_{l}^{-1/4})\quad(l=1,2).
\]
However, achieving such a uniform convergence rate typically requires additional conditions, such as compact covariate support or covariate densities bounded away from zero. 

\begin{theorem}\label{thm:null}
    Under $H_{0}$ in (\ref{eq:null}) and Assumptions \ref{assu1}-\ref{assu3}, we have
    \[
    \left(\frac{\sigma_{1}^{2}}{n_{1}}+\frac{\sigma_{2}^{2}}{n_{2}}\right)^{-1/2}T\convd\Nsc(0,1),\quad\text{as }n_{1},n_{2}\rightarrow\infty,
    \]
    where $\sigma_{l}^{2}=\E[\{\varepsilon^{(l)}\varepsilon^{(l)\prime}k(X^{(l)},X^{(l)\prime})\}^{2}]$ and $(\varepsilon^{(l)\prime},X^{(l)\prime})$ is an independent copy of $(\varepsilon^{(l)},X^{(l)})$. 
\end{theorem}
It is worth noting that Theorem \ref{thm:null} remains valid whether the dimension $p$ is fixed or diverges as $n_{1},n_{2}\rightarrow\infty$. A natural plug-in estimator of $\sigma_{l}^{2}$ is given by
\[
\sigmahat_{l}^{2}=\frac{1}{n_{l}}\sum_{i=1}^{n_{l}}\left\{\etahat^{(l)}_{i}\etahat^{(l)}_{i+n_{l}}k(X^{(l)}_{i},X^{(l)}_{i+n_{l}})\right\}^{2}\quad(l=1,2), 
\]
The following theorem shows that $\sigmahat_{l}^{2}$ is ratio-consistent under the null hypothesis. Then we reject $H_{0}$ at a significance level $\alpha$ if $(\sigmahat_{1}^{2}/n_{1}+\sigmahat_{2}^{2}/n_{2})^{-1/2}T>z_{1-\alpha}$. 
\begin{theorem}\label{thm:ratio}
    Under $H_{0}$ in (\ref{eq:null}) and Assumptions \ref{assu1}-\ref{assu3}, we have $\sigmahat_{l}^{2}/\sigma_{l}^{2}\convp 1$ for $l=1,2$. Consequently, $(\sigmahat_{1}^{2}/n_{1}+\sigmahat_{2}^{2}/n_{2})^{-1/2}T\convd\Nsc(0,1)$ as $n_{1},n_{2}\rightarrow\infty$. 
\end{theorem}

\begin{remark}\label{rmk:ustat}
    One might consider constructing a test statistic based on the U-statistic estimator
    \[
    \Deltahat^{(l)}_{U}=\frac{1}{n_{l}(n_{l}-1)}\sum_{1\le i\neq j\le n_{l}}\etahat^{(l)}_{i}\etahat^{(l)}_{j}k(X^{(l)}_{i},X^{(l)}_{j})\quad(l=1,2).
    \]
    We instead adopt $\Deltahat^{(l)}$ for two principal reasons. 
    First, employing $\Deltahat^{(l)}_{U}$ requires strengthening the convergence rate in Assumption \ref{assu3} from $o_{p}(n_{l}^{-1/2})$ to $o_{p}(n_{l}^{-1})$—a condition that fails even under simple parametric models. Without this stronger rate, it becomes necessary to use a substantially smaller test set. For example, Assumption 3 and Theorem 1 in \citet{he2025goodness} require the test set size to be of order $o(\sqrt{n_{l}})$. 
    Second, the asymptotic distribution of $\Deltahat^{(l)}_{U}$ differs markedly between fixed- and high-dimensional regimes. When $p$ is fixed, $\Deltahat^{(l)}_{U}$ is a degenerate U-statistic whose null distribution is an infinite weighted sum of chi-squared variables, where the weights depend on the unknown underlying distribution. This necessitates a resampling calibration procedure such as the wild bootstrap \citep{lai2021kernel}, which incurs considerable computational cost. In contrast, when $p$ diverges, a martingale CLT argument \citep{hall2014martingale} yields the asymptotic normality of $\Deltahat^{(l)}_{U}$ under Assumption \ref{assu2} and the additional condition
    \[
    \frac{\E\left(\left[\E\left\{h(Z_{1}^{(l)},Z_{2}^{(l)})h(Z_{1}^{(l)},Z_{3}^{(l)})\mid Z_{2}^{(l)},Z_{3}^{(l)}\right\}\right]^{2}\right)}{\left[\E\left\{h^{2}(Z_{1}^{(l)},Z_{2}^{(l)})\right\}\right]^{2}}\longrightarrow0,\quad\text{as }n_{l}\rightarrow\infty,
    \]
    where $Z_{i}^{(l)}=(\varepsilon_{i}^{(l)},X_{i}^{(l)})$ and $h(Z_{1}^{(l)},Z_{2}^{(l)})=\varepsilon^{(l)}_{1}\varepsilon^{(l)}_{2}k(X^{(l)}_{1},X^{(l)}_{2})$. Similar conditions appear in \citet{zhang2018conditional,li2023testing}, but they are generally unverifiable without strong prior knowledge. By contrast, Theorem \ref{thm:null} indicates that $\Deltahat^{(l)}$ is asymptotically normal under the null hypothesis in both fixed- and high-dimensional settings, without the need for resampling procedures or reliance on such practically unverifiable technical conditions. 
\end{remark}

Next, we study the asymptotic behavior of $T$ under the alternative hypothesis. 
\begin{theorem}\label{thm:alt}
    For $l=1,2$, assume that $\E[\{m^{(1)}(X^{(l)})-m^{(2)}(X^{(l)})\}^{2}]\le C_{3}<\infty$ for some $C_{3}>0$. Under $H_{a}$ in (\ref{eq:null}) and Assumptions \ref{assu1}-\ref{assu3}, we have
    \begin{align*}
        T&=\Delta^{(1)}+\Delta^{(2)}+O_{p}(n_{1}^{-1/2}+n_{2}^{-1/2})\\
        &\quad+\|m^{(1)}-m^{(2)}\|_{L^{2}(P_{X}^{(2)})}o_{p}(n_{1}^{-1/4})+\|m^{(1)}-m^{(2)}\|_{L^{2}(P_{X}^{(1)})}o_{p}(n_{2}^{-1/4}).
    \end{align*}
    Thus, sufficient conditions for consistency, i.e., $(\sigmahat_{1}^{2}/n_{1}+\sigmahat_{2}^{2}/n_{2})^{-1/2}T\convp\infty$, are $\Delta^{(1)}+\Delta^{(2)}=\omega(n_{1}^{-1/2}+n_{2}^{-1/2})$ and $\Delta^{(1)}+\Delta^{(2)}=\Omega(\|m^{(1)}-m^{(2)}\|_{L^{2}(P_{X}^{(2)})}n_{1}^{-1/4}+\|m^{(1)}-m^{(2)}\|_{L^{2}(P_{X}^{(1)})}n_{2}^{-1/4})$. 
\end{theorem}

\begin{remark}
    The slower rates $o_{p}(n_{1}^{-1/4})$ and $o_{p}(n_{2}^{-1/4})$ under $H_{a}$ arise from the first-order errors of the machine learning estimators $\mhat^{(l)}(\cdot)$. By contrast, under $H_{0}$, only second-order errors contribute to the asymptotics. Under parametric models, these rates improve to $O_{p}(n_{1}^{-1/2})$ and $O_{p}(n_{2}^{-1/2})$. 
\end{remark}

Although the test statistic $T$ is valid across all dimensional regimes, its power exhibits distinct behaviors in fixed- and high-dimensional settings. When $p$ is fixed and $\Delta^{(1)}+\Delta^{(2)}$ is a constant under $H_{a}$, the test is consistent, i.e., $(\sigmahat_{1}^{2}/n_{1}+\sigmahat_{2}^{2}/n_{2})^{-1/2}T\convp\infty$. In contrast, when $p\rightarrow\infty$, the quantity $\Delta^{(1)}+\Delta^{(2)}$ depends on $p$, and the power is determined jointly by the dimension and the sample sizes. To illustrate, consider the Gaussian kernel $k(x,x')=\exp(-(2\gamma^{2})^{-1}\|x-x'\|^{2})$, where $\gamma$ is the bandwidth parameter. Following \citet{yan2023kernel,han2024generalized}, we assume $\gamma^{2}\asymp p$ and impose the following condition. 
\begin{assumption}\label{assu4}
    Let $X^{(1)}=\mu_{1}+\Sigma_{1}^{1/2}U$ and $X^{(2)}=\mu_{2}+\Sigma_{2}^{1/2}V$, where $\mu_{l}$ and $\Sigma_{l}$ denote the mean vector and covariance matrix of $X^{(l)}$, respectively. The random vectors $U,V\in\R^{p}$ have independent components with mean zero, variance one and finite fourth moment. Moreover, the spectrum of $\Sigma_{l}$ lies within $[M^{-1},M]$ for some $M>1$. 
\end{assumption}

Assumption \ref{assu4} is standard in high-dimensional statistics and random matrix theory. We emphasize that this assumption is introduced solely for theoretical analysis and is not required for the implementation of our test. The bandwidth order $\gamma^{2}\asymp p$ is consistent with the widely used median heuristic \citep{ramdas2015adaptivity}. 

\begin{theorem}\label{thm:high}
    For $l=1,2$, let $\tau_{l}=\tr(\Sigma_{l})/\gamma^{2}$, where $\tr(\cdot)$ denotes the trace. Under $H_{a}$ in (\ref{eq:null}) and Assumption \ref{assu4}, we have
    \[
    \Delta^{(1)}+\Delta^{(2)}=e^{-\tau_{1}}\left\{\E(\eta^{(1)})\right\}^{2}+e^{-\tau_{2}}\left\{\E(\eta^{(2)})\right\}^{2}+\left\{\E(|\eta^{(1)}|^{2})+\E(|\eta^{(2)}|^{2})\right\}O(p^{-1/2}).
    \]
    When $\E(\eta^{(l)})=0$ for $l=1,2$, we further have
    \begin{align*}
        \Delta^{(1)}+\Delta^{(2)}=\frac{e^{-\tau_{1}}}{\gamma^{2}}\left\|\cov(\eta^{(1)},X^{(1)})\right\|^{2}+\frac{e^{-\tau_{2}}}{\gamma^{2}}\left\|\cov(\eta^{(2)},X^{(2)})\right\|^{2}+\left\{\E(|\eta^{(1)}|^{2})+\E(|\eta^{(2)}|^{2})\right\}O(p^{-1}).
    \end{align*}
\end{theorem}
Combining Theorems \ref{thm:alt} and \ref{thm:high}, we obtain the following conclusion. If
\[
\max\left\{\left|\E(\eta^{(1)})\right|,\left|\E(\eta^{(2)})\right|\right\}\ge c_{3},
\]
for some $c_{3}>0$, the proposed test is consistent regardless of the growth rate of the dimension $p$. By contrast, when $\E(\eta^{(l)})=0$ for both $l=1,2$, the test statistic $T$ primarily detects alternatives satisfying $\cov(\eta^{(l)},X^{(l)})\neq 0$ for some $l=1,2$, as $p\rightarrow\infty$. To capture the nonlinear dependence of $\eta^{(l)}$ on $X^{(l)}$ in high-dimensional settings, we introduce an alternative test statistic in Section \ref{sec:alt}. 
\begin{remark}
    By the definition of $\eta^{(l)}\ (l=1,2)$, we have
    \begin{align*}
        \E(\eta^{(1)})&=\E\left\{Y^{(1)}-m^{(2)}(X^{(1)})\right\}=\E\left\{m^{(1)}(X^{(1)})-m^{(2)}(X^{(1)})\right\},\\
        \E(\eta^{(2)})&=\E\left\{Y^{(2)}-m^{(1)}(X^{(2)})\right\}=\E\left\{m^{(2)}(X^{(2)})-m^{(1)}(X^{(2)})\right\}.
    \end{align*}
    In the simplest case where both $m^{(1)}(\cdot)$ and $m^{(2)}(\cdot)$ are linear, the equalities $\E(\eta^{(1)})=\E(\eta^{(2)})=0$ hold under $H_{a}$ provided that $\mu_{1}=\mu_{2}=0$ and the two models share the same intercept. Generally, however, the regression functions are often highly nonlinear in practice. Therefore, under the alternative hypothesis, it is uncommon for both $\E(\eta^{(1)})$ and $\E(\eta^{(2)})$ to vanish simultaneously, especially in the presence of covariate shift (i.e., $P_{X}^{(1)}\neq P_{X}^{(2)}$). 
\end{remark}

\section{Alternative test statistic in high dimensions}\label{sec:alt}
For $l=1,2$, write $X^{(l)}=(X^{(l)}(1),\ldots,X^{(l)}(p))^{\top}$ and $X^{(l)}_{i}=(X^{(l)}_{i}(1),\ldots,X^{(l)}_{i}(p))^{\top}\ (i=1,\ldots,N_{l})$. 
Given a kernel $k$ on $\R$, define
\[
\Delta^{(l)}_{d}=\text{KCMD}^{*}(\eta^{(l)}\mid X^{(l)}(d))=\E\left\{\eta^{(l)}\eta^{(l)\prime}k(X^{(l)}(d),X^{(l)\prime}(d))\right\}\quad(l=1,2;\ d=1,\ldots,p),
\]
where $(\eta^{(l)\prime},X^{(l)\prime})$ is an independent copy of $(\eta^{(l)},X^{(l)})$. 

Now we propose an alternative test statistic for detecting nonlinear conditional mean dependence of $\eta^{(l)}$ on $X^{(l)}$ in high dimensions:
\[
T_{a}=\sum_{l=1}^{2}\sum_{d=1}^{p}\Deltahat^{(l)}_{d},
\]
where
\[
\Deltahat^{(l)}_{d}=\frac{1}{n_{l}}\sum_{i=1}^{n_{l}}\etahat^{(l)}_{i}\etahat^{(l)}_{i+n_{l}}k(X^{(l)}_{i}(d),X^{(l)}_{i+n_{l}}(d)),
\]
and $\etahat^{(l)}_{i}$ is defined as in Section \ref{sec:test}. 

In contrast to (\ref{eq:equiv}), which is equivalent to the null hypothesis $H_{0}$, the test statistic $T_{a}$ targets the weaker null hypothesis
\begin{equation}\label{eq:weak}
    \E\{\eta^{(l)}\mid X^{(l)}(d)\}=0\quad(l=1,2;\ d=1,\ldots,p)\quad\text{almost surely}.
\end{equation}
Similar idea was first introduced by \citet{zhang2018conditional} and later adopted by \citet{li2023testing}. This relaxation is necessary in high-dimensional regimes, where the alternative space is vast due to both the growing dimension and the possibility of nonlinear dependence. 

\begin{remark}
    Our problem differs entirely from that of \citet{zhang2018conditional,li2023testing}. Furthermore, our goal is to test whether $\E(\eta^{(l)}\mid X^{(l)})=0$, rather than whether $\E(\eta^{(l)}\mid X^{(l)})=\E(\eta^{(l)})$. To illustrate, suppose that $m^{(1)}(\cdot)-m^{(2)}(\cdot)\equiv a$ for some $a\neq 0$; that is, the two regression functions differ only by a constant shift. In this case, we have $\E(\eta^{(l)}\mid X^{(l)})\neq 0$ but $\E(\eta^{(l)}\mid X^{(l)})=\E(\eta^{(l)})$. 
\end{remark}

We next establish the asymptotic properties of $T_{a}$. To this end, Assumption \ref{assu2} is reformulated as follows.
\begin{assumptionp}{assu2}\label{assu2p}
    Define 
    \[
    G(X^{(l)},X^{(l)\prime})=\sum_{d=1}^{p}k(X^{(l)}(d),X^{(l)\prime}(d))\quad(l=1,2). 
    \]
    For $l=1,2$, assume that $\E(|\varepsilon^{(l)}|^{4})\le C_{1}<\infty$ for some $C_{1}>0$, and
    \begin{equation}\label{eq:assu2p}
        \E\left\{\left|\varepsilon^{(l)}\varepsilon^{(l)\prime}G(X^{(l)},X^{(l)\prime})\right|^{2}\right\}\ge c_{4}p^{2},
    \end{equation}
    for some $c_{4}>0$, where $(\varepsilon^{(l)\prime},X^{(l)\prime})$ denotes an independent copy of $(\varepsilon^{(l)},X^{(l)})$. 
\end{assumptionp}
Condition (\ref{eq:assu2p}) holds automatically when $p$ is fixed. For diverging $p$, a sufficient condition for (\ref{eq:assu2p}) is that
\[
\min_{1\le d_{1},d_{2}\le p}\E\left\{\left|\varepsilon^{(l)}\varepsilon^{(l)\prime}\right|^{2}k(X^{(l)}(d_{1}),X^{(l)\prime}(d_{1}))k(X^{(l)}(d_{2}),X^{(l)\prime}(d_{2}))\right\}\ge c_{4},
\]
for some $c_{4}>0$.

\begin{theorem}\label{thm:nulla}
    Under $H_{0}$ in (\ref{eq:null}) and Assumptions \ref{assu1}, \ref{assu2p} and \ref{assu3}, we have
    \[
    \left(\frac{\xi_{1}^{2}}{n_{1}}+\frac{\xi_{2}^{2}}{n_{2}}\right)^{-1/2}T_{a}\convd\Nsc(0,1),\quad\text{as }n_{1},n_{2}\rightarrow\infty,
    \]
    where $\xi_{l}^{2}=\E[\{\varepsilon^{(l)}\varepsilon^{(l)\prime}G(X^{(l)},X^{(l)\prime})\}^{2}]$ and $(\varepsilon^{(l)\prime},X^{(l)\prime})$ is an independent copy of $(\varepsilon^{(l)},X^{(l)})$. 
\end{theorem}
\begin{remark}
    Heuristically, the statistic $T_{a}$ can be regarded as a version of the statistic $T$ with kernel $G$ defined in Assumption \ref{assu2p}. Accordingly, the theoretical analysis proceeds in a manner analogous to that presented in Section \ref{sec:test}. 
\end{remark}

A natural plug-in estimator of $\xi_{l}^{2}$ is given by
\[
\xihat_{l}^{2}=\frac{1}{n_{l}}\sum_{i=1}^{n_{l}}\left\{\etahat^{(l)}_{i}\etahat^{(l)}_{i+n_{l}}\sum_{d=1}^{p}k(X^{(l)}_{i}(d),X^{(l)}_{i+n_{l}}(d))\right\}^{2}\quad(l=1,2). 
\]
The following theorem shows that $\xihat_{l}^{2}$ is ratio-consistent under $H_{0}$. Then we reject $H_{0}$ at a significance level $\alpha$ if $(\xihat_{1}^{2}/n_{1}+\xihat_{2}^{2}/n_{2})^{-1/2}T_{a}>z_{1-\alpha}$. 
\begin{theorem}\label{thm:ratioa}
    Under $H_{0}$ in (\ref{eq:null}) and Assumptions \ref{assu1}, \ref{assu2p} and \ref{assu3}, we have $\xihat_{l}^{2}/\xi_{l}^{2}\convp 1$ for $l=1,2$. As a result, $(\xihat_{1}^{2}/n_{1}+\xihat_{2}^{2}/n_{2})^{-1/2}T_{a}\convd\Nsc(0,1)$ as $n_{1},n_{2}\rightarrow\infty$. 
\end{theorem}

Although the statistic $T_{a}$ is initially designed for high-dimensional settings, Theorems \ref{thm:nulla} and \ref{thm:ratioa} remain valid whether the dimension $p$ is fixed or diverging. In the fixed-dimensional case, however, the test based on $T$ introduced in Section \ref{sec:test} is preferable, as it is consistent against all alternatives, whereas $T_{a}$ can merely detect alternatives for which the weaker null hypothesis (\ref{eq:weak}) is violated. 

\begin{theorem}\label{thm:alta}
    For $l=1,2$, assume that $\E[\{m^{(1)}(X^{(l)})-m^{(2)}(X^{(l)})\}^{2}]\le C_{3}<\infty$ for some $C_{3}>0$. Under $H_{a}$ in (\ref{eq:null}) and Assumptions \ref{assu1}, \ref{assu2p} and \ref{assu3}, we have
    \begin{align*}
        T_{a}&=\sum_{d=1}^{p}\Delta^{(1)}_{d}+\sum_{d=1}^{p}\Delta^{(2)}_{d}+O_{p}(n_{1}^{-1/2}p+n_{2}^{-1/2}p)\\
        &\quad+\|m^{(1)}-m^{(2)}\|_{L^{2}(P_{X}^{(2)})}o_{p}(n_{1}^{-1/4}p)+\|m^{(1)}-m^{(2)}\|_{L^{2}(P_{X}^{(1)})}o_{p}(n_{2}^{-1/4}p).
    \end{align*}
    Hence, sufficient conditions for consistency, i.e., $(\xihat_{1}^{2}/n_{1}+\xihat_{2}^{2}/n_{2})^{-1/2}T_{a}\convp\infty$, are $\sum_{d=1}^{p}\Delta^{(1)}_{d}+\sum_{d=1}^{p}\Delta^{(2)}_{d}=\omega(n_{1}^{-1/2}p+n_{2}^{-1/2}p)$ and $\sum_{d=1}^{p}\Delta^{(1)}_{d}+\sum_{d=1}^{p}\Delta^{(2)}_{d}=\Omega(\|m^{(1)}-m^{(2)}\|_{L^{2}(P_{X}^{(2)})}n_{1}^{-1/4}p+\|m^{(1)}-m^{(2)}\|_{L^{2}(P_{X}^{(1)})}n_{2}^{-1/4}p)$. 
\end{theorem}

When $\max\{|\E(\eta^{(1)})|,|\E(\eta^{(2)})|\}\ge c_{3}$ for some $c_{3}>0$, then $\sum_{d=1}^{p}\Delta^{(1)}_{d}+\sum_{d=1}^{p}\Delta^{(2)}_{d}\asymp p$. Consequently, similar to the test based on $T$, the test constructed from $T_{a}$ remains consistent regardless of the growth rate of the dimension $p$. Moreover, when the dimension $p$ increases and $\E(\eta^{(l)})=0$ for both $l=1,2$, the statistic $T_{a}$, through the quantity $\sum_{d=1}^{p}\Delta^{(1)}_{d}+\sum_{d=1}^{p}\Delta^{(2)}_{d}$, is able to capture componentwise nonlinear dependence between $\eta^{(l)}$ and $X^{(l)}$, in contrast to the test based on $T$. 
\begin{remark}
    In practice, without prior knowledge of the underlying dependence structure, both tests can be applied irrespective of the dimensionality. The extensive numerical studies demonstrate the reliable performance of both test statistics. 
\end{remark}


\section{Numerical studies}\label{sec:num}
\subsection{Monte Carlo simulations}
In this section, we conduct several Monte Carlo simulations to evaluate the finite-sample performance of our proposed tests, denoted by $T$, the statistic from Section \ref{sec:test}, and $T_a$, the alternative statistic from Section \ref{sec:alt}.
We assess their empirical size and power across various settings, including low-dimensional and high-dimensional covariates, and different forms of nonlinear regression functions.
All results are based on 1000 replications for each setting, and the nominal significance level is set to $\alpha = 0.05$.

For our proposed methods, we employ XGBoost \citep{chen2016xgboost} to estimate the regression functions $m^{(1)}(\cdot)$ and $m^{(2)}(\cdot)$.
We opt for XGBoost owing to its strong predictive performance across diverse data structures, its computational efficiency, and its relative simplicity in tuning compared with more complex methods like neural networks, which typically require larger sample sizes and more extensive hyperparameter optimization.
We set a learning rate of $0.15$ and a maximum tree depth of $3$. To prevent both underfitting and overfitting, the optimal number of boosting iterations is determined using 5-fold cross-validation. At the same time, we incorporate an early stopping mechanism, which terminates the training process if the validation performance does not improve for 20 consecutive rounds, with an overall limit of $1000$ iterations.
For both test statistics, we use a Gaussian kernel, $k(x, x') = \exp(-(2\gamma^2)^{-1} \| x - x' \|^2)$.
The bandwidth parameter $\gamma$ is selected via the median heuristic \citep{gretton2012kernel}, and is computed separately for each sample.
Specifically, for the statistic $T$, a scalar bandwidth is calculated for each sample based on the median of pairwise Euclidean distances within that sample.
For the statistic $T_a$, a dimension-specific bandwidth vector is computed for each sample, where each component is derived from the median of pairwise distances along the corresponding covariate dimension within that sample.

Across all examples, we set the sample sizes to be equal, $N_1 = N_2 = N$, and consider different values of $N$ in the simulations.
For each group $l \in \{1, 2\}$, we consider the regression model
\begin{equation*}
    Y^{(l)} = m^{(l)}(X^{(l)}) + \varepsilon^{(l)}.
\end{equation*}
The covariates for the first group are drawn from a standard normal distribution, $X^{(1)} \sim \Nsc(0, I_p)$.
To introduce a covariate shift, the covariates for the second group are drawn from a normal distribution with an autoregressive covariance structure, $X^{(2)} \sim \Nsc(0, \Sigma)$, where $\Sigma_{ij} = 0.3^{|i - j|}$.
The errors are generated as $\varepsilon^{(1)} \sim \Nsc(0, 0.5^2)$ and $\varepsilon^{(2)}$ from a t-distribution with $5$ degrees of freedom, scaled to have a variance of $0.5^2$.
The difference between the two regression functions is controlled by a vector $\beta\in\R^{p}$, where the null hypothesis $H_0$ corresponds to $\|\beta\| = 0$. 

\begin{example}\label{exam1}
    Write $x=(x(1),\ldots,x(p))^{\top}\in\R^{p}$. Let
    \begin{align*}
        m^{(1)}(x) &= \sqrt{\{x(1)\}^2 + \{x(2)\}^2 + \sum_{d=1}^{p} \beta_{d} \{x(d)\}^2}, \\
        m^{(2)}(x) &= \sqrt{\{x(1)\}^2 + \{x(2)\}^2}.
    \end{align*}
\end{example}

\begin{example}\label{exam2}
    Let
    \begin{align*}
        m^{(1)}(x) &= \exp\{ x(1) \} + \exp\{ x(2) \} + \sum_{d=1}^{p} \beta_d \exp\{ x(d) \}, \\
        m^{(2)}(x) &= \exp\{ x(1) \} + \exp\{ x(2) \}.
    \end{align*}
\end{example}

\begin{example}\label{exam3}
    Let
    \begin{align*}
        m^{(1)}(x) &= x(1) + x(2) + \left\{ \sum_{d=1}^{p} \beta_d x(d) \right\}^2, \\
        m^{(2)}(x) &= x(1) + x(2).
    \end{align*}
\end{example}

For each example, we consider a low-dimensional setting with $p = 5$ and a high-dimensional setting with $p = 50$.
In the low-dimensional case, all the elements of signal vector $\beta$ are non-zero and have equal magnitude.
In the high-dimensional case, we investigate two types of alternatives: a \textit{sparse} alternative, where only the first two elements of $\beta$ are non-zero with equal magnitude, and a \textit{dense} alternative, where the first $20$ elements of $\beta$ are non-zero, also with equal magnitude.

For the low-dimensional case ($p = 5$), our tests are compared with the smoothing-based procedure of \citet{lavergne2001equality}. 
Consistent with their simulation study, a uniform kernel is employed, and the bandwidth is selected via the rule-of-thumb method.

The simulation results are presented in Tables \ref{tab:eg1}--\ref{tab:eg3}.
Under the null hypothesis (i.e., $\|\beta\| = 0$), the empirical sizes of our proposed tests, $T$ and $T_a$, are consistently close to the $5\%$ nominal level across all settings.
This indicates that our tests maintain proper size control even in the presence of covariate shift, non-Gaussian errors, and high dimensionality. 

As expected, the power of both tests, $T$ and $T_a$, increases with the signal strength $\|\beta\|$ and the sample size $N$, approaching one as either becomes sufficiently large. In high-dimensional settings ($p=50$), both procedures perform well against both sparse and dense alternatives. By contrast, even in low-dimensional settings ($p=5$), the test of \citet{lavergne2001equality} is overly conservative, exhibiting an empirical size of zero across all scenarios. 
As a result, it displays no empirical power and fails to detect the nonlinear alternatives in any of the three examples. 

\begin{table}[htbp]
    \centering
    \caption{Empirical size and power for Example \ref{exam1}.}\label{tab:eg1}
    \begin{tabular}{cccccccccc}
    \toprule
            &          &          & \multicolumn{1}{c}{$T$} & \multicolumn{1}{c}{$T_a$} & \multicolumn{1}{c}{Lavergne} &          & \multicolumn{1}{c}{$T$} & \multicolumn{1}{c}{$T_a$} & \multicolumn{1}{c}{Lavergne} \\
    \midrule
            & \multicolumn{1}{c}{$\|\beta\|$} &          & \multicolumn{3}{c}{$N=200$}    &          & \multicolumn{3}{c}{$N=400$} \\
    \cmidrule{4-6}\cmidrule{8-10}\multirow{6}[2]{*}{$p=5$} & 0        &          & 0.063    & 0.064    & 0.000    &          & 0.048    & 0.049    & 0.000  \\
            & 0.06     &          & 0.093    & 0.094    & 0.000    &          & 0.110    & 0.110    & 0.000  \\
            & 0.12     &          & 0.172    & 0.182    & 0.000    &          & 0.261    & 0.275    & 0.000  \\
            & 0.18     &          & 0.309    & 0.309    & 0.000    &          & 0.487    & 0.504    & 0.001  \\
            & 0.24     &          & 0.512    & 0.516    & 0.000    &          & 0.742    & 0.765    & 0.000  \\
            & 0.3      &          & 0.684    & 0.705    & 0.000    &          & 0.914    & 0.925    & 0.000  \\
    \midrule
            & \multicolumn{1}{c}{$\|\beta\|$} &          & \multicolumn{3}{c}{$N=500$}    &          & \multicolumn{3}{c}{$N=1000$} \\
    \cmidrule{4-6}\cmidrule{8-10}\multicolumn{1}{c}{\multirow{6}[2]{*}{$p=50$, dense}} & 0        &          & 0.059    & 0.059    &          &          & 0.055    & 0.060    &  \\
            & 0.02     &          & 0.077    & 0.082    &          &          & 0.098    & 0.102    &  \\
            & 0.04     &          & 0.178    & 0.195    &          &          & 0.246    & 0.252    &  \\
            & 0.06     &          & 0.366    & 0.362    &          &          & 0.574    & 0.574    &  \\
            & 0.08     &          & 0.571    & 0.578    &          &          & 0.841    & 0.845    &  \\
            & 0.1      &          & 0.808    & 0.813    &          &          & 0.962    & 0.964    &  \\
    \midrule
            & \multicolumn{1}{c}{$\|\beta\|$} &          & \multicolumn{3}{c}{$N=500$}    &          & \multicolumn{3}{c}{$N=1000$} \\
    \cmidrule{4-6}\cmidrule{8-10}\multicolumn{1}{c}{\multirow{6}[2]{*}{$p=50$, sparse}} & 0        &          & 0.059    & 0.059    &          &          & 0.055    & 0.060    &  \\
            & 0.1      &          & 0.094    & 0.092    &          &          & 0.099    & 0.101    &  \\
            & 0.2      &          & 0.184    & 0.180    &          &          & 0.241    & 0.245    &  \\
            & 0.3      &          & 0.319    & 0.319    &          &          & 0.548    & 0.556    &  \\
            & 0.4      &          & 0.560    & 0.566    &          &          & 0.846    & 0.851    &  \\
            & 0.5      &          & 0.783    & 0.787    &          &          & 0.975    & 0.974    &  \\
    \bottomrule
    \end{tabular}%
\end{table}


\begin{table}[htbp]
    \centering
    \caption{Empirical size and power for Example \ref{exam2}.}\label{tab:eg2}
    \begin{tabular}{cccccccccc}
    \toprule
            &          &          & \multicolumn{1}{c}{$T$} & \multicolumn{1}{c}{$T_a$} & \multicolumn{1}{c}{Lavergne} &          & \multicolumn{1}{c}{$T$} & \multicolumn{1}{c}{$T_a$} & \multicolumn{1}{c}{Lavergne} \\
    \midrule
            & \multicolumn{1}{c}{$\|\beta\|$} &          & \multicolumn{3}{c}{$N=200$}    &          & \multicolumn{3}{c}{$N=400$} \\
    \cmidrule{4-6}\cmidrule{8-10}\multirow{6}[2]{*}{$p=5$} & 0        &          & 0.055    & 0.052    & 0.000    &          & 0.052    & 0.056    & 0.000  \\
            & 0.02     &          & 0.065    & 0.054    & 0.000    &          & 0.072    & 0.069    & 0.000  \\
            & 0.04     &          & 0.103    & 0.098    & 0.001    &          & 0.146    & 0.126    & 0.000  \\
            & 0.06     &          & 0.186    & 0.168    & 0.000    &          & 0.349    & 0.300    & 0.000  \\
            & 0.08     &          & 0.319    & 0.292    & 0.000    &          & 0.582    & 0.536    & 0.000  \\
            & 0.1      &          & 0.504    & 0.450    & 0.000    &          & 0.836    & 0.785    & 0.000  \\
    \midrule
            & \multicolumn{1}{c}{$\|\beta\|$} &          & \multicolumn{3}{c}{$N=500$}    &          & \multicolumn{3}{c}{$N=1000$} \\
    \cmidrule{4-6}\cmidrule{8-10}\multicolumn{1}{c}{\multirow{6}[2]{*}{$p=50$, dense}} & 0        &          & 0.047    & 0.051    &          &          & 0.059    & 0.055    &  \\
            & 0.02     &          & 0.091    & 0.099    &          &          & 0.164    & 0.160    &  \\
            & 0.04     &          & 0.485    & 0.477    &          &          & 0.781    & 0.769    &  \\
            & 0.06     &          & 0.857    & 0.859    &          &          & 0.986    & 0.981    &  \\
            & 0.08     &          & 0.970    & 0.965    &          &          & 0.994    & 0.995    &  \\
            & 0.1      &          & 0.991    & 0.990    &          &          & 0.999    & 0.999    &  \\
    \midrule
            & \multicolumn{1}{c}{$\|\beta\|$} &          & \multicolumn{3}{c}{$N=500$}    &          & \multicolumn{3}{c}{$N=1000$} \\
    \cmidrule{4-6}\cmidrule{8-10}\multicolumn{1}{c}{\multirow{6}[2]{*}{$p=50$, sparse}} & 0        &          & 0.047    & 0.051    &          &          & 0.059    & 0.055    &  \\
            & 0.04     &          & 0.055    & 0.058    &          &          & 0.061    & 0.059    &  \\
            & 0.08     &          & 0.171    & 0.162    &          &          & 0.299    & 0.296    &  \\
            & 0.12     &          & 0.425    & 0.415    &          &          & 0.731    & 0.727    &  \\
            & 0.16     &          & 0.685    & 0.672    &          &          & 0.935    & 0.929    &  \\
            & 0.2      &          & 0.869    & 0.869    &          &          & 0.984    & 0.984    &  \\
    \bottomrule
    \end{tabular}%
\end{table}

\begin{table}[htbp]
    \centering
    \caption{Empirical size and power for Example \ref{exam3}.}\label{tab:eg3}
    \begin{tabular}{cccccccccc}
    \toprule
            &          &          & \multicolumn{1}{c}{$T$} & \multicolumn{1}{c}{$T_a$} & \multicolumn{1}{c}{Lavergne} &          & \multicolumn{1}{c}{$T$} & \multicolumn{1}{c}{$T_a$} & \multicolumn{1}{c}{Lavergne} \\
    \midrule
            & \multicolumn{1}{c}{$\|\beta\|$} &          & \multicolumn{3}{c}{$N=200$}    &          & \multicolumn{3}{c}{$N=400$} \\
    \cmidrule{4-6}\cmidrule{8-10}\multirow{6}[2]{*}{$p=5$} & 0        &          & 0.066    & 0.067    & 0.000    &          & 0.066    & 0.061    & 0.000  \\
            & 0.12     &          & 0.084    & 0.074    & 0.000    &          & 0.056    & 0.060    & 0.000  \\
            & 0.24     &          & 0.098    & 0.090    & 0.000    &          & 0.092    & 0.096    & 0.000  \\
            & 0.36     &          & 0.195    & 0.193    & 0.000    &          & 0.286    & 0.289    & 0.000  \\
            & 0.48     &          & 0.447    & 0.456    & 0.000    &          & 0.790    & 0.807    & 0.000  \\
            & 0.6      &          & 0.874    & 0.880    & 0.000    &          & 0.992    & 0.991    & 0.001  \\
    \midrule
            & \multicolumn{1}{c}{$\|\beta\|$} &          & \multicolumn{3}{c}{$N=500$}    &          & \multicolumn{3}{c}{$N=1000$} \\
    \cmidrule{4-6}\cmidrule{8-10}\multicolumn{1}{c}{\multirow{6}[2]{*}{$p=50$, dense}} & 0        &          & 0.068    & 0.064    &          &          & 0.051    & 0.052    &  \\
            & 0.1      &          & 0.058    & 0.060    &          &          & 0.065    & 0.063    &  \\
            & 0.2      &          & 0.077    & 0.075    &          &          & 0.086    & 0.084    &  \\
            & 0.3      &          & 0.157    & 0.157    &          &          & 0.232    & 0.237    &  \\
            & 0.4      &          & 0.423    & 0.422    &          &          & 0.727    & 0.726    &  \\
            & 0.5      &          & 0.877    & 0.877    &          &          & 0.994    & 0.992    &  \\
    \midrule
            & \multicolumn{1}{c}{$\|\beta\|$} &          & \multicolumn{3}{c}{$N=500$}    &          & \multicolumn{3}{c}{$N=1000$} \\
    \cmidrule{4-6}\cmidrule{8-10}\multicolumn{1}{c}{\multirow{6}[2]{*}{$p=50$, sparse}} & 0        &          & 0.068    & 0.064    &          &          & 0.051    & 0.052    &  \\
            & 0.1      &          & 0.058    & 0.060    &          &          & 0.069    & 0.068    &  \\
            & 0.2      &          & 0.072    & 0.075    &          &          & 0.102    & 0.101    &  \\
            & 0.3      &          & 0.166    & 0.161    &          &          & 0.274    & 0.277    &  \\
            & 0.4      &          & 0.482    & 0.481    &          &          & 0.779    & 0.789    &  \\
            & 0.5      &          & 0.905    & 0.906    &          &          & 0.995    & 0.995    &  \\
    \bottomrule
    \end{tabular}%
\end{table}

\subsection{Airfoil data example}
We apply our proposed tests to the Airfoil Self-Noise dataset \citep{airfoil}, available from the UCI Machine Learning Repository.
This dataset, which has also been used in previous studies such as \citet{tibshirani2019conformal,hu2024two}, includes $1503$ measurements from wind tunnel experiments on NASA airfoil sections.
The response is the scaled sound pressure level, and there are five covariates: frequency, angle of attack, chord length, free-stream velocity, and suction side displacement thickness.
Following common practice, we apply a logarithmic transformation to the frequency and thickness variables.

Since the original dataset does not have a predefined two-sample structure, we construct two scenarios to evaluate our tests.
First, to assess the empirical size, we simulate the null hypothesis by randomly partitioning the entire dataset into two nearly equal groups with sizes $N_1=752$ and $N_2=751$.
Under this setup, the underlying regression functions for the two groups are expected to be identical. 
We repeat this procedure $500$ times to estimate the Type I error rate.
Second, to evaluate the power, we construct an alternative hypothesis by partitioning the data based on the response. 
Specifically, one group comprises observations with sound pressure levels at or below the median, while the other contains the remaining observations. 
This procedure intentionally creates two groups with distinct regression functions. 
To avoid trivial separation of the covariate supports, we randomly swap $5\%$ of the samples between the two groups. 

The results of our analysis demonstrate the effectiveness of the proposed tests.
Under the null hypothesis, where the data is randomly partitioned, the empirical rejection rates for the statistics $T$ and $T_a$ are $4.6\%$ and $5.8\%$, respectively.
Both rates are close to the nominal $5\%$ level, confirming that our tests maintain proper size control on this real-world dataset.
Under the alternative hypothesis, constructed by splitting the data based on the response, both tests strongly reject the null hypothesis.
The test statistic $T$ yields a $p$-value of $1.574 \times 10^{-19}$, while the statistic $T_a$ gives a $p$-value of $4.726 \times 10^{-17}$.

\subsection{Communities and Crime data example}



We further validate our tests on the Communities and Crime dataset \citep{communities_and_crime}, also from the UCI Machine Learning Repository.
This dataset, which has also been used by \citet{joshi2025conformal}, combines socio-economic data from the 1990 US Census, law enforcement data from the 1990 US LEMAS survey, and crime data from the 1995 FBI Uniform Crime Report.
The dataset contains 1994 communities and 128 attributes.
The response is the per capita violent crime rate, and covariates include information like population, median income, and employment rates.
After removing non-predictive attributes and any covariates containing missing values, we proceed with the analysis. 
This preprocessing results in a final covariate dimension of $p = 99$.

To create a two-sample testing problem, we again construct two scenarios following the same logic as in the Airfoil data analysis. 
To evaluate the empirical size, we create a null scenario by randomly splitting the dataset into two samples of equal size, a procedure repeated 500 times. 
To assess the power, we generate an alternative hypothesis by dividing the communities into two groups based on their crime rates (at or below vs. above the median) and then swapping $5\%$ of the data points between the groups to ensure overlapping covariate distributions.

Our tests perform well in this setting.
Under the null hypothesis, the empirical rejection rates are $5.6\%$ for the statistic $T$ and $5.2\%$ for the statistic $T_a$, both of which are close to the nominal $5\%$ level, demonstrating accurate size control.
Under the alternative hypothesis, the tests show decisive power, yielding $p$-values of $2.128 \times 10^{-76}$ for $T$ and $9.260 \times 10^{-83}$ for $T_a$.

\section{Discussion}
To conclude, we outline several directions for future research. First, it would be valuable to extend the framework to the comparison of $L$ regression functions for a general $L\ge 2$, or even for $L\rightarrow\infty$. Second, beyond the conditional mean, it would be of interest to employ machine learning methods to test the equality of two conditional variance functions \citep{pardo2015tests}, or more generally, two conditional distributions \citep{yan2022distance}. Finally, our framework may find utility in other contexts, such as change point detection \citep{nie2022detection}.

\newpage
\begin{appendices}
\counterwithin{theorem}{section}
\counterwithin{definition}{section}
\renewcommand{\theequation}{A.\arabic{equation}}
\setcounter{equation}{0}

\section{Technical details} 
\begin{proof}[Proof of Lemma \ref{lemma:kcmd}]
    For any $v_{0}\in\R$, we have
    \begin{align*}
        \text{KCMD}^{*}(V\mid U)&=\E\{(V-v_{0})(V'-v_{0})k(U,U')\}\\
        &=\E\{\langle(V-v_{0})k(\cdot,U),(V'-v_{0})k(\cdot,U')\rangle_{\Hsc_{k}}\}\\
        &=\|\E\{(V-v_{0})k(\cdot,U)\}\|_{\Hsc_{k}}^{2}\ge 0. 
    \end{align*}
    Also, $\text{KCMD}^{*}(V\mid U)=0$ if and only if $\E\{(V-v_{0})k(\cdot,U)\}$ is the zero element in $\Hsc_{k}$. 

    Let $\sigma(U)$ be the $\sigma$-algebra generated by $U$. Define a finite signed Borel measure $\lambda$ on $\sigma(U)$
    \[
    \lambda(B)=\E\{(V-v_{0})\mathbbm{1}_{B}\},\quad\forall B\in\sigma(U).
    \]
    Then the kernel mean embedding of $\lambda\circ U^{-1}$ into the RKHS $\Hsc_{k}$ is exactly $\E\{(V-v_{0})k(\cdot,U)\}$. When the kernel $k$ is characteristic, $\E\{(V-v_{0})k(\cdot,U)\}$ is the zero element in $\Hsc_{k}$ if and only if $\lambda$ is the zero measure on $\sigma(U)$. Thus, $\text{KCMD}^{*}(V\mid U)=0$ if and only if
    \[
    \E\{(V-v_{0})\mathbbm{1}_{B}\}=0,\quad\forall B\in\sigma(U),
    \]
    which holds if and only if 
    \[
    \E(V\mid U)=v_{0}\quad\text{almost surely},
    \]
    by the definition of conditional expectation.
\end{proof}

\begin{proof}[Proof of Theorem \ref{thm:kcmd}]
    The results directly follow from Lemma \ref{lemma:kcmd} with $v_{0}=0$.  
\end{proof}

\begin{proof}[Proof of Theorem \ref{thm:null}]
    Under $H_{0}$, we have
    \begin{align*}
        \etahat^{(1)}_{i}&=Y^{(1)}_{i}-\mhat^{(2)}(X^{(1)}_{i})\\
        &=Y^{(1)}_{i}-m^{(2)}(X^{(1)}_{i})+m^{(2)}(X^{(1)}_{i})-\mhat^{(2)}(X^{(1)}_{i})\\
        &=Y^{(1)}_{i}-m^{(1)}(X^{(1)}_{i})+m^{(2)}(X^{(1)}_{i})-\mhat^{(2)}(X^{(1)}_{i})\\
        &=\varepsilon^{(1)}_{i}+m^{(2)}(X^{(1)}_{i})-\mhat^{(2)}(X^{(1)}_{i}),
    \end{align*}
    and thus
    \begin{align*}
        \Deltahat^{(1)}&=\frac{1}{n_{1}}\sum_{i=1}^{n_{1}}\left\{\varepsilon^{(1)}_{i}+m^{(2)}(X^{(1)}_{i})-\mhat^{(2)}(X^{(1)}_{i})\right\}\left\{\varepsilon^{(1)}_{i+n_{1}}+m^{(2)}(X^{(1)}_{i+n_{1}})-\mhat^{(2)}(X^{(1)}_{i+n_{1}})\right\}k(X^{(1)}_{i},X^{(1)}_{i+n_{1}})\\
        &=\frac{1}{n_{1}}\sum_{i=1}^{n_{1}}\varepsilon^{(1)}_{i}\varepsilon^{(1)}_{i+n_{1}}k(X^{(1)}_{i},X^{(1)}_{i+n_{1}})+S_{1}+S_{2}+S_{3},
    \end{align*}
    where
    \begin{align*}
        S_{1}&=-\frac{1}{n_{1}}\sum_{i=1}^{n_{1}}\varepsilon^{(1)}_{i}\left\{\mhat^{(2)}(X^{(1)}_{i+n_{1}})-m^{(2)}(X^{(1)}_{i+n_{1}})\right\}k(X^{(1)}_{i},X^{(1)}_{i+n_{1}}),\\
        S_{2}&=-\frac{1}{n_{1}}\sum_{i=1}^{n_{1}}\left\{\mhat^{(2)}(X^{(1)}_{i})-m^{(2)}(X^{(1)}_{i})\right\}\varepsilon^{(1)}_{i+n_{1}}k(X^{(1)}_{i},X^{(1)}_{i+n_{1}}),\\
        S_{3}&=\frac{1}{n_{1}}\sum_{i=1}^{n_{1}}\left\{\mhat^{(2)}(X^{(1)}_{i})-m^{(2)}(X^{(1)}_{i})\right\}\left\{\mhat^{(2)}(X^{(1)}_{i+n_{1}})-m^{(2)}(X^{(1)}_{i+n_{1}})\right\}k(X^{(1)}_{i},X^{(1)}_{i+n_{1}}).
    \end{align*}
    Under Assumptions \ref{assu1}-\ref{assu3},
    \begin{align*}
        \E(S_{1}\mid\Dsc^{(2)})&=-\E\left[\varepsilon^{(1)}\left\{\mhat^{(2)}(X^{(1)\prime})-m^{(2)}(X^{(1)\prime})\right\}k(X^{(1)},X^{(1)\prime})\mid\Dsc^{(2)}\right]\\
        &=-\E\left[\E(\varepsilon^{(1)}\mid X^{(1)})\left\{\mhat^{(2)}(X^{(1)\prime})-m^{(2)}(X^{(1)\prime})\right\}k(X^{(1)},X^{(1)\prime})\mid\Dsc^{(2)}\right]\\
        &=0,\\
        \var(S_{1}\mid\Dsc^{(2)})&=\frac{1}{n_{1}}\var\left[\varepsilon^{(1)}\left\{\mhat^{(2)}(X^{(1)\prime})-m^{(2)}(X^{(1)\prime})\right\}k(X^{(1)},X^{(1)\prime})\mid\Dsc^{(2)}\right]\\
        &\le\frac{1}{n_{1}}\E\left[\left|\varepsilon^{(1)}\left\{\mhat^{(2)}(X^{(1)\prime})-m^{(2)}(X^{(1)\prime})\right\}k(X^{(1)},X^{(1)\prime})\right|^{2}\mid\Dsc^{(2)}\right]\\
        &\le\frac{1}{n_{1}}K^{2}\E(|\varepsilon^{(1)}|^{2})\E\left\{\left|\mhat^{(2)}(X^{(1)\prime})-m^{(2)}(X^{(1)\prime})\right|^{2}\mid\Dsc^{(2)}\right\}\\
        &=o_{p}(n_{1}^{-1}n_{2}^{-1/2}),\\
        \E(S_{2}\mid\Dsc^{(2)})&=\E(S_{1}\mid\Dsc^{(2)})=0,\\
        \var(S_{2}\mid\Dsc^{(2)})&=\var(S_{1}\mid\Dsc^{(2)})=o_{p}(n_{1}^{-1}n_{2}^{-1/2}),\\
        \E(|S_{3}|\mid\Dsc^{(2)})&\le\E\left\{\left|\mhat^{(2)}(X^{(1)})-m^{(2)}(X^{(1)})\right|\left|\mhat^{(2)}(X^{(1)\prime})-m^{(2)}(X^{(1)\prime})\right||k(X^{(1)},X^{(1)\prime})|\mid\Dsc^{(2)}\right\}\\
        &\le K\left[\E\left\{\left|\mhat^{(2)}(X^{(1)})-m^{(2)}(X^{(1)})\right|\mid\Dsc^{(2)}\right\}\right]^{2}\\
        &\le K\E\left\{\left|\mhat^{(2)}(X^{(1)})-m^{(2)}(X^{(1)})\right|^{2}\mid\Dsc^{(2)}\right\}\\
        &=o_{p}(n_{2}^{-1/2}).
    \end{align*}
    Using Lemma 6.1 in \citet{chernozhukov2018double}, we have
    \begin{align*}
        S_{1}&=o_{p}(n_{1}^{-1/2}n_{2}^{-1/4}),\\
        S_{2}&=o_{p}(n_{1}^{-1/2}n_{2}^{-1/4}),\\
        S_{3}&=o_{p}(n_{2}^{-1/2}).
    \end{align*}
    
    A similar analysis can be carried out for $\Deltahat^{(2)}$. Consequently, under $H_{0}$,
    \[
    T=\sum_{l=1}^{2}\frac{1}{n_{l}}\sum_{i=1}^{n_{l}}\varepsilon^{(l)}_{i}\varepsilon^{(l)}_{i+n_{l}}k(X^{(l)}_{i},X^{(l)}_{i+n_{l}})+o_{p}(n_{1}^{-1/2}+n_{2}^{-1/2}).
    \]
    Applying the Lyapunov CLT, we have
    \[
    \left(\frac{\sigma_{1}^{2}}{n_{1}}+\frac{\sigma_{2}^{2}}{n_{2}}\right)^{-1/2}T=\left(\frac{\sigma_{1}^{2}}{n_{1}}+\frac{\sigma_{2}^{2}}{n_{2}}\right)^{-1/2}\sum_{l=1}^{2}\frac{1}{n_{l}}\sum_{i=1}^{n_{l}}\varepsilon^{(l)}_{i}\varepsilon^{(l)}_{i+n_{l}}k(X^{(l)}_{i},X^{(l)}_{i+n_{l}})+o_{p}(1)\convd\Nsc(0,1),
    \]
    where we use Assumptions \ref{assu1} and \ref{assu2} to verify the Lyapunov condition:
    \[
    \lim_{n_{1},n_{2}\rightarrow\infty}\frac{\sum_{l=1}^{2}n_{l}^{-3}\E\left\{\left|\varepsilon^{(l)}\varepsilon^{(l)\prime}k(X^{(l)},X^{(l)\prime})\right|^{4}\right\}}{\left[\sum_{l=1}^{2}n_{l}^{-1}\E\left\{\left|\varepsilon^{(l)}\varepsilon^{(l)\prime}k(X^{(l)},X^{(l)\prime})\right|^{2}\right\}\right]^{2}}\le\lim_{n_{1},n_{2}\rightarrow\infty}\frac{C_{1}^{2}K^{4}\sum_{l=1}^{2}n_{l}^{-3}}{c_{1}^{2}\sum_{l=1}^{2}n_{l}^{-2}}=0. 
    \]
\end{proof}

\begin{proof}[Proof of Theorem \ref{thm:ratio}]
    Under $H_{0}$, we have
    \begin{align*}
        \sigmahat_{1}^{2}&=\frac{1}{n_{1}}\sum_{i=1}^{n_{1}}\left[\left\{\varepsilon^{(1)}_{i}+m^{(2)}(X^{(1)}_{i})-\mhat^{(2)}(X^{(1)}_{i})\right\}\left\{\varepsilon^{(1)}_{i+n_{1}}+m^{(2)}(X^{(1)}_{i+n_{1}})-\mhat^{(2)}(X^{(1)}_{i+n_{1}})\right\}k(X^{(1)}_{i},X^{(1)}_{i+n_{1}})\right]^{2}\\
        &=\frac{1}{n_{1}}\sum_{i=1}^{n_{1}}\left\{\varepsilon^{(1)}_{i}\varepsilon^{(1)}_{i+n_{1}}k(X^{(1)}_{i},X^{(1)}_{i+n_{1}})\right\}^{2}+Q_{1}+Q_{2},
    \end{align*}
    where
    \begin{align*}
        Q_{1}&=\frac{1}{n_{1}}\sum_{i=1}^{n_{1}}\Big[-\varepsilon^{(1)}_{i}\left\{\mhat^{(2)}(X^{(1)}_{i+n_{1}})-m^{(2)}(X^{(1)}_{i+n_{1}})\right\}-\left\{\mhat^{(2)}(X^{(1)}_{i})-m^{(2)}(X^{(1)}_{i})\right\}\varepsilon^{(1)}_{i+n_{1}}\\
        &\quad+\left\{\mhat^{(2)}(X^{(1)}_{i})-m^{(2)}(X^{(1)}_{i})\right\}\left\{\mhat^{(2)}(X^{(1)}_{i+n_{1}})-m^{(2)}(X^{(1)}_{i+n_{1}})\right\}\Big]^{2}\left\{k(X^{(1)}_{i},X^{(1)}_{i+n_{1}})\right\}^{2},\\
        Q_{2}&=\frac{2}{n_{1}}\sum_{i=1}^{n_{1}}\varepsilon^{(1)}_{i}\varepsilon^{(1)}_{i+n_{1}}\Big[-\varepsilon^{(1)}_{i}\left\{\mhat^{(2)}(X^{(1)}_{i+n_{1}})-m^{(2)}(X^{(1)}_{i+n_{1}})\right\}-\left\{\mhat^{(2)}(X^{(1)}_{i})-m^{(2)}(X^{(1)}_{i})\right\}\varepsilon^{(1)}_{i+n_{1}}\\
        &\quad+\left\{\mhat^{(2)}(X^{(1)}_{i})-m^{(2)}(X^{(1)}_{i})\right\}\left\{\mhat^{(2)}(X^{(1)}_{i+n_{1}})-m^{(2)}(X^{(1)}_{i+n_{1}})\right\}\Big]\left\{k(X^{(1)}_{i},X^{(1)}_{i+n_{1}})\right\}^{2}. 
    \end{align*}
    By similar calculations as in the proof of Theorem \ref{thm:null}, one can derive
    \begin{align*}
        \E(|Q_{1}|\mid\Dsc^{(2)})&=o_{p}(n_{2}^{-1/2}),\\
        \E(|Q_{2}|\mid\Dsc^{(2)})&=o_{p}(n_{2}^{-1/4}).
    \end{align*}
    Using Lemma 6.1 in \citet{chernozhukov2018double},
    \begin{align*}
        Q_{1}&=o_{p}(n_{2}^{-1/2}),\\
        Q_{2}&=o_{p}(n_{2}^{-1/4}).
    \end{align*}
    As 
    \[
    \frac{1}{n_{1}}\sum_{i=1}^{n_{1}}\left\{\varepsilon^{(1)}_{i}\varepsilon^{(1)}_{i+n_{1}}k(X^{(1)}_{i},X^{(1)}_{i+n_{1}})\right\}^{2}=\sigma_{1}^{2}+O_{p}(n_{1}^{-1/2}),
    \]
    and $\sigma_{1}^{2}\asymp 1$ under Assumptions \ref{assu1} and \ref{assu2}, we have 
    \[
    \frac{\sigmahat_{1}^{2}-\sigma_{1}^{2}}{\sigma_{1}^{2}}\convp 0.
    \] 
    
    A similar analysis can be carried out for $\sigmahat_{2}^{2}$. The results follow.
\end{proof}

\begin{proof}[Proof of Theorem \ref{thm:alt}]
    Under $H_{a}$, we have
    \begin{align*}
        \etahat^{(1)}_{i}&=\eta^{(1)}_{i}+m^{(2)}(X^{(1)}_{i})-\mhat^{(2)}(X^{(1)}_{i})\\
        &=\varepsilon^{(1)}_{i}+m^{(1)}(X^{(1)}_{i})-m^{(2)}(X^{(1)}_{i})+m^{(2)}(X^{(1)}_{i})-\mhat^{(2)}(X^{(1)}_{i}),
    \end{align*}
    and thus
    \begin{align*}
        \Deltahat^{(1)}&=\frac{1}{n_{1}}\sum_{i=1}^{n_{1}}\left\{\eta^{(1)}_{i}+m^{(2)}(X^{(1)}_{i})-\mhat^{(2)}(X^{(1)}_{i})\right\}\left\{\eta^{(1)}_{i+n_{1}}+m^{(2)}(X^{(1)}_{i+n_{1}})-\mhat^{(2)}(X^{(1)}_{i+n_{1}})\right\}k(X^{(1)}_{i},X^{(1)}_{i+n_{1}})\\
        &=\frac{1}{n_{1}}\sum_{i=1}^{n_{1}}\eta^{(1)}_{i}\eta^{(1)}_{i+n_{1}}k(X^{(1)}_{i},X^{(1)}_{i+n_{1}})+S_{1}+S_{2}+S_{3}+S_{4}+S_{5},
    \end{align*}
    where $S_{1},S_{2},S_{3}$ are defined as in the proof of Theorem \ref{thm:null}, and
    \begin{align*}
        S_{4}&=-\frac{1}{n_{1}}\sum_{i=1}^{n_{1}}\left\{m^{(1)}(X^{(1)}_{i})-m^{(2)}(X^{(1)}_{i})\right\}\left\{\mhat^{(2)}(X^{(1)}_{i+n_{1}})-m^{(2)}(X^{(1)}_{i+n_{1}})\right\}k(X^{(1)}_{i},X^{(1)}_{i+n_{1}}),\\
        S_{5}&=-\frac{1}{n_{1}}\sum_{i=1}^{n_{1}}\left\{\mhat^{(2)}(X^{(1)}_{i})-m^{(2)}(X^{(1)}_{i})\right\}\left\{m^{(1)}(X^{(1)}_{i+n_{1}})-m^{(2)}(X^{(1)}_{i+n_{1}})\right\}k(X^{(1)}_{i},X^{(1)}_{i+n_{1}}).
    \end{align*}
    Following the proof of Theorem \ref{thm:null}, we still have
    \begin{align*}
        S_{1}&=o_{p}(n_{1}^{-1/2}n_{2}^{-1/4}),\\
        S_{2}&=o_{p}(n_{1}^{-1/2}n_{2}^{-1/4}),\\
        S_{3}&=o_{p}(n_{2}^{-1/2}).
    \end{align*}
    Moreover, under Assumptions \ref{assu1} and \ref{assu3}, 
    \begin{align*}
        \E(|S_{4}|\mid\Dsc^{(2)})&\le\E\left\{\left|m^{(1)}(X^{(1)})-m^{(2)}(X^{(1)})\right|\left|\mhat^{(2)}(X^{(1)\prime})-m^{(2)}(X^{(1)\prime})\right||k(X^{(1)},X^{(1)\prime})|\mid\Dsc^{(2)}\right\}\\
        &\le K\sqrt{\E\left\{\left|m^{(1)}(X^{(1)})-m^{(2)}(X^{(1)})\right|^{2}\right\}}\sqrt{\E\left\{\left|\mhat^{(2)}(X^{(1)\prime})-m^{(2)}(X^{(1)\prime})\right|^{2}\mid\Dsc^{(2)}\right\}}\\
        &=\|m^{(1)}-m^{(2)}\|_{L^{2}(P_{X}^{(1)})}o_{p}(n_{2}^{-1/4}).
    \end{align*}
    Using Lemma 6.1 in \citet{chernozhukov2018double}, we have
    \[
    S_{4}=\|m^{(1)}-m^{(2)}\|_{L^{2}(P_{X}^{(1)})}o_{p}(n_{2}^{-1/4}).
    \]
    Likewise, one can derive $S_{5}=\|m^{(1)}-m^{(2)}\|_{L^{2}(P_{X}^{(1)})}o_{p}(n_{2}^{-1/4})$. 
    
    A similar analysis can be carried out for $\Deltahat^{(2)}$. Consequently, under $H_{a}$,
    \begin{align*}
        T&=\sum_{l=1}^{2}\frac{1}{n_{l}}\sum_{i=1}^{n_{l}}\eta^{(l)}_{i}\eta^{(l)}_{i+n_{l}}k(X^{(l)}_{i},X^{(l)}_{i+n_{l}})+o_{p}(n_{1}^{-1/2}+n_{2}^{-1/2})\\
        &\quad+\|m^{(1)}-m^{(2)}\|_{L^{2}(P_{X}^{(2)})}o_{p}(n_{1}^{-1/4})+\|m^{(1)}-m^{(2)}\|_{L^{2}(P_{X}^{(1)})}o_{p}(n_{2}^{-1/4})\\
        &=\Delta^{(1)}+\Delta^{(2)}+O_{p}(n_{1}^{-1/2}+n_{2}^{-1/2})\\
        &\quad+\|m^{(1)}-m^{(2)}\|_{L^{2}(P_{X}^{(2)})}o_{p}(n_{1}^{-1/4})+\|m^{(1)}-m^{(2)}\|_{L^{2}(P_{X}^{(1)})}o_{p}(n_{2}^{-1/4}).
    \end{align*}
    As $\sigmahat_{l}^{2}=\E[\{\eta^{(l)}\eta^{(l)\prime}k(X^{(l)},X^{(l)\prime})\}^{2}]\{1+o_{p}(1)\}$ for $l=1,2$, the results follow.
\end{proof}

\begin{proof}[Proof of Theorem \ref{thm:high}]
    Using the second-order Taylor expansion, for any $\tau\ge 0$,
    \[
    k(x,x')=e^{-\frac{\|x-x'\|^{2}}{2\gamma^{2}}}=e^{-\tau}-e^{-\tau}\left(\frac{\|x-x'\|^{2}}{2\gamma^{2}}-\tau\right)+\frac{e^{-\xi}}{2}\left(\frac{\|x-x'\|^{2}}{2\gamma^{2}}-\tau\right)^{2},
    \]
    where $\xi$ is some real number between $\tau$ and $(2\gamma^{2})^{-1}\|x-x'\|^{2}$. Thus, we have
    \[
    \Delta^{(1)}=\E\left\{\eta^{(1)}\eta^{(1)\prime}k(X^{(1)},X^{(1)\prime})\right\}=R_{0}+R_{1}+R_{2},
    \]
    where
    \begin{align*}
        R_{0}&=e^{-\tau_{1}}\E(\eta^{(1)}\eta^{(1)\prime})=e^{-\tau_{1}}\left\{\E(\eta^{(1)})\right\}^{2},\\
        R_{1}&=-e^{-\tau_{1}}\E\left\{\eta^{(1)}\eta^{(1)\prime}\left(\frac{\|X^{(1)}-X^{(1)\prime}\|^{2}}{2\gamma^{2}}-\tau_{1}\right)\right\}\\
        &=-\frac{e^{-\tau_{1}}}{2\gamma^{2}}\E\left[\eta^{(1)}\eta^{(1)\prime}\left\{\|X^{(1)}-X^{(1)\prime}\|^{2}-2\tr(\Sigma_{1})\right\}\right],\\
        R_{2}&=\frac{1}{2}\E\left\{\eta^{(1)}\eta^{(1)\prime}e^{-\xi_{1}}\left(\frac{\|X^{(1)}-X^{(1)\prime}\|^{2}}{2\gamma^{2}}-\tau_{1}\right)^{2}\right\}\\
        &=\frac{1}{8\gamma^{4}}\E\left[\eta^{(1)}\eta^{(1)\prime}e^{-\xi_{1}}\left\{\|X^{(1)}-X^{(1)\prime}\|^{2}-2\tr(\Sigma_{1})\right\}^{2}\right],
    \end{align*}
    and $\xi_{1}$ is some random variable between $\tau_{1}$ and $(2\gamma^{2})^{-1}\|X^{(1)}-X^{(1)\prime}\|^{2}$. 

    Under Assumption \ref{assu4}, we have
    \begin{align*}
        |R_{1}|&\le\frac{e^{-\tau_{1}}}{2\gamma^{2}}\E\left\{|\eta^{(1)}||\eta^{(1)\prime}|\left|\|X^{(1)}-X^{(1)\prime}\|^{2}-2\tr(\Sigma_{1})\right|\right\}\\
        &\le\frac{e^{-\tau_{1}}}{2\gamma^{2}}\sqrt{\E(|\eta^{(1)}|^{2}|\eta^{(1)\prime}|^{2})}\sqrt{\E\left[\left\{\|X^{(1)}-\mu_{1}+\mu_{1}-X^{(1)\prime}\|^{2}-2\tr(\Sigma_{1})\right\}^{2}\right]}\\
        &=\frac{e^{-\tau_{1}}}{2\gamma^{2}}\E(|\eta^{(1)}|^{2})\sqrt{\E\left[\left\{U^{\top}\Sigma_{1}U-\tr(\Sigma_{1})+U^{\prime\top}\Sigma_{1}U'-\tr(\Sigma_{1})-2U^{\top}\Sigma_{1}U'\right\}^{2}\right]}\\
        &\overset{(i)}{=}O(\gamma^{-2})\E(|\eta^{(1)}|^{2})\sqrt{\E\left[\left\{U^{\top}\Sigma_{1}U-\tr(\Sigma_{1})\right\}^{2}\right]+\E\left\{(U^{\top}\Sigma_{1}U')^{2}\right\}}\\
        &\overset{(ii)}{=}O(\gamma^{-2})\E(|\eta^{(1)}|^{2})\sqrt{\tr(\Sigma_{1}^{2})}\\
        &=O(p^{-1/2})\E(|\eta^{(1)}|^{2}),\\
        |R_{2}|&\le\frac{1}{8\gamma^{4}}\E\left[|\eta^{(1)}||\eta^{(1)\prime}|\left\{\|X^{(1)}-X^{(1)\prime}\|^{2}-2\tr(\Sigma_{1})\right\}^{2}\right]\\
        &\le\frac{1}{8\gamma^{4}}\E(|\eta^{(1)}|^{2})\sqrt{\E\left[\left\{U^{\top}\Sigma_{1}U-\tr(\Sigma_{1})+U^{\prime\top}\Sigma_{1}U'-\tr(\Sigma_{1})-2U^{\top}\Sigma_{1}U'\right\}^{4}\right]}\\
        &\overset{(iii)}{=}O(\gamma^{-4})\E(|\eta^{(1)}|^{2})\sqrt{\E\left[\left\{U^{\top}\Sigma_{1}U-\tr(\Sigma_{1})\right\}^{4}\right]+\E\left\{(U^{\top}\Sigma_{1}U')^{4}\right\}}\\
        &\overset{(iv)}{=}O(\gamma^{-4})\E(|\eta^{(1)}|^{2})\tr(\Sigma_{1}^{2})\\
        &=O(p^{-1})\E(|\eta^{(1)}|^{2}),
    \end{align*}
    where (i) and (iii) hold by Jensen's inequality, and (ii) and (iv) use Lemmas S4 and S5 in \citet{yan2023kernel}. 

    When $\E(\eta^{(1)})=0$, we have $R_{0}=0$ and $R_{1}$ reduces to
    \begin{align*}
        R_{1}&=-\frac{e^{-\tau_{1}}}{2\gamma^{2}}\E\left[\eta^{(1)}\eta^{(1)\prime}\left\{\|X^{(1)}\|^{2}+\|X^{(1)\prime}\|^{2}-2X^{(1)\top}X^{(1)\prime}-2\tr(\Sigma_{1})\right\}\right]\\
        &=-\frac{e^{-\tau_{1}}}{\gamma^{2}}\left\{\E\left(\eta^{(1)}\|X^{(1)}\|^{2}\right)\E(\eta^{(1)\prime})-\E\left(\eta^{(1)}\eta^{(1)\prime}X^{(1)\top}X^{(1)\prime}\right)-\tr(\Sigma_{1})\E(\eta^{(1)})\E(\eta^{(1)\prime})\right\}\\
        &=\frac{e^{-\tau_{1}}}{\gamma^{2}}\left\|\cov(\eta^{(1)},X^{(1)})\right\|^{2}.
    \end{align*}
    
    A similar analysis can be carried out for $\Delta^{(2)}$. The results follow. 
\end{proof}

\begin{proof}[Proof of Theorem \ref{thm:nulla}]
    Under $H_{0}$, we have
    \[
    \etahat^{(1)}_{i}=\varepsilon^{(1)}_{i}+m^{(2)}(X^{(1)}_{i})-\mhat^{(2)}(X^{(1)}_{i}),
    \]
    and thus
    \[
    \sum_{d=1}^{p}\Deltahat^{(1)}_{d}=\frac{1}{n_{1}}\sum_{i=1}^{n_{1}}\varepsilon^{(1)}_{i}\varepsilon^{(1)}_{i+n_{1}}\sum_{d=1}^{p}k(X^{(1)}_{i}(d),X^{(1)}_{i+n_{1}}(d))+S_{1}+S_{2}+S_{3},
    \]
    where
    \begin{align*}
        S_{1}&=-\frac{1}{n_{1}}\sum_{i=1}^{n_{1}}\varepsilon^{(1)}_{i}\left\{\mhat^{(2)}(X^{(1)}_{i+n_{1}})-m^{(2)}(X^{(1)}_{i+n_{1}})\right\}\sum_{d=1}^{p}k(X^{(1)}_{i}(d),X^{(1)}_{i+n_{1}}(d)),\\
        S_{2}&=-\frac{1}{n_{1}}\sum_{i=1}^{n_{1}}\left\{\mhat^{(2)}(X^{(1)}_{i})-m^{(2)}(X^{(1)}_{i})\right\}\varepsilon^{(1)}_{i+n_{1}}\sum_{d=1}^{p}k(X^{(1)}_{i}(d),X^{(1)}_{i+n_{1}}(d)),\\
        S_{3}&=\frac{1}{n_{1}}\sum_{i=1}^{n_{1}}\left\{\mhat^{(2)}(X^{(1)}_{i})-m^{(2)}(X^{(1)}_{i})\right\}\left\{\mhat^{(2)}(X^{(1)}_{i+n_{1}})-m^{(2)}(X^{(1)}_{i+n_{1}})\right\}\sum_{d=1}^{p}k(X^{(1)}_{i}(d),X^{(1)}_{i+n_{1}}(d)).
    \end{align*}
    
    Recall $G(X^{(l)},X^{(l)\prime})=\sum_{d=1}^{p}k(X^{(l)}(d),X^{(l)\prime}(d))\ (l=1,2)$, which satisfies
    \[
    \left|G(X^{(l)},X^{(l)\prime})\right|\le\sum_{d=1}^{p}\left|k(X^{(l)}(d),X^{(l)\prime}(d))\right|=Kp,
    \]
    where we use Assumption \ref{assu1}. Then, under Assumptions \ref{assu1}, \ref{assu2p} and \ref{assu3},
    \begin{align*}
        \E(S_{1}\mid\Dsc^{(2)})&=-\E\left[\varepsilon^{(1)}\left\{\mhat^{(2)}(X^{(1)\prime})-m^{(2)}(X^{(1)\prime})\right\}G(X^{(1)},X^{(1)\prime})\mid\Dsc^{(2)}\right]\\
        &=-\E\left[\E(\varepsilon^{(1)}\mid X^{(1)})\left\{\mhat^{(2)}(X^{(1)\prime})-m^{(2)}(X^{(1)\prime})\right\}G(X^{(1)},X^{(1)\prime})\mid\Dsc^{(2)}\right]\\
        &=0,\\
        \var(S_{1}\mid\Dsc^{(2)})&=\frac{1}{n_{1}}\var\left[\varepsilon^{(1)}\left\{\mhat^{(2)}(X^{(1)\prime})-m^{(2)}(X^{(1)\prime})\right\}G(X^{(1)},X^{(1)\prime})\mid\Dsc^{(2)}\right]\\
        &\le\frac{1}{n_{1}}\E\left[\left|\varepsilon^{(1)}\left\{\mhat^{(2)}(X^{(1)\prime})-m^{(2)}(X^{(1)\prime})\right\}G(X^{(1)},X^{(1)\prime})\right|^{2}\mid\Dsc^{(2)}\right]\\
        &\le\frac{1}{n_{1}}(Kp)^{2}\E(|\varepsilon^{(1)}|^{2})\E\left\{\left|\mhat^{(2)}(X^{(1)\prime})-m^{(2)}(X^{(1)\prime})\right|^{2}\mid\Dsc^{(2)}\right\}\\
        &=o_{p}(n_{1}^{-1}n_{2}^{-1/2}p^{2}),\\
        \E(S_{2}\mid\Dsc^{(2)})&=\E(S_{1}\mid\Dsc^{(2)})=0,\\
        \var(S_{2}\mid\Dsc^{(2)})&=\var(S_{1}\mid\Dsc^{(2)})=o_{p}(n_{1}^{-1}n_{2}^{-1/2}p^{2}),\\
        \E(|S_{3}|\mid\Dsc^{(2)})&\le\E\left\{\left|\mhat^{(2)}(X^{(1)})-m^{(2)}(X^{(1)})\right|\left|\mhat^{(2)}(X^{(1)\prime})-m^{(2)}(X^{(1)\prime})\right||G(X^{(1)},X^{(1)\prime})|\mid\Dsc^{(2)}\right\}\\
        &\le Kp\left[\E\left\{\left|\mhat^{(2)}(X^{(1)})-m^{(2)}(X^{(1)})\right|\mid\Dsc^{(2)}\right\}\right]^{2}\\
        &\le Kp\E\left\{\left|\mhat^{(2)}(X^{(1)})-m^{(2)}(X^{(1)})\right|^{2}\mid\Dsc^{(2)}\right\}\\
        &=o_{p}(n_{2}^{-1/2}p).
    \end{align*}
    Using Lemma 6.1 in \citet{chernozhukov2018double}, we have
    \begin{align*}
        S_{1}&=o_{p}(n_{1}^{-1/2}n_{2}^{-1/4}p),\\
        S_{2}&=o_{p}(n_{1}^{-1/2}n_{2}^{-1/4}p),\\
        S_{3}&=o_{p}(n_{2}^{-1/2}p).
    \end{align*}
    
    A similar analysis can be carried out for $\sum_{d=1}^{p}\Deltahat^{(2)}_{d}$. Consequently, under $H_{0}$, 
    \[
    T_{a}=\sum_{l=1}^{2}\frac{1}{n_{l}}\sum_{i=1}^{n_{l}}\varepsilon^{(l)}_{i}\varepsilon^{(l)}_{i+n_{l}}\sum_{d=1}^{p}k(X^{(l)}_{i}(d),X^{(l)}_{i+n_{l}}(d))+o_{p}(n_{1}^{-1/2}p+n_{2}^{-1/2}p).
    \]
    Under Assumptions \ref{assu1} and \ref{assu2p}, we have
    \[
    \xi_{l}^{2}\asymp p^{2}\quad(l=1,2),
    \]
    and thus
    \[
    \left(\frac{\xi_{1}^{2}}{n_{1}}+\frac{\xi_{2}^{2}}{n_{2}}\right)^{-1/2}T_{a}=\left(\frac{\xi_{1}^{2}}{n_{1}}+\frac{\xi_{2}^{2}}{n_{2}}\right)^{-1/2}\sum_{l=1}^{2}\frac{1}{n_{l}}\sum_{i=1}^{n_{l}}\varepsilon^{(l)}_{i}\varepsilon^{(l)}_{i+n_{l}}\sum_{d=1}^{p}k(X^{(l)}_{i}(d),X^{(l)}_{i+n_{l}}(d))+o_{p}(1). 
    \]
    Applying the Lyapunov CLT, 
    \[
    \left(\frac{\xi_{1}^{2}}{n_{1}}+\frac{\xi_{2}^{2}}{n_{2}}\right)^{-1/2}T_{a}\convd\Nsc(0,1),
    \]
    where we use Assumptions \ref{assu1} and \ref{assu2p} to verify the Lyapunov condition:
    \[
    \lim_{n_{1},n_{2}\rightarrow\infty}\frac{\sum_{l=1}^{2}n_{l}^{-3}\E\left\{\left|\varepsilon^{(l)}\varepsilon^{(l)\prime}G(X^{(l)},X^{(l)\prime})\right|^{4}\right\}}{\left[\sum_{l=1}^{2}n_{l}^{-1}\E\left\{\left|\varepsilon^{(l)}\varepsilon^{(l)\prime}G(X^{(l)},X^{(l)\prime})\right|^{2}\right\}\right]^{2}}\le\lim_{n_{1},n_{2}\rightarrow\infty}\frac{C_{1}^{2}K^{4}\sum_{l=1}^{2}n_{l}^{-3}}{c_{4}^{2}\sum_{l=1}^{2}n_{l}^{-2}}=0.
    \]
\end{proof}

\begin{proof}[Proof of Theorem \ref{thm:ratioa}]
    Under $H_{0}$, we have
    \[
    \xihat_{1}^{2}=\frac{1}{n_{1}}\sum_{i=1}^{n_{1}}\left\{\varepsilon^{(1)}_{i}\varepsilon^{(1)}_{i+n_{1}}\sum_{d=1}^{p}k(X^{(1)}_{i}(d),X^{(1)}_{i+n_{1}}(d))\right\}^{2}+Q_{1}+Q_{2},
    \]
    where
    \begin{align*}
        Q_{1}&=\frac{1}{n_{1}}\sum_{i=1}^{n_{1}}\Big[-\varepsilon^{(1)}_{i}\left\{\mhat^{(2)}(X^{(1)}_{i+n_{1}})-m^{(2)}(X^{(1)}_{i+n_{1}})\right\}-\left\{\mhat^{(2)}(X^{(1)}_{i})-m^{(2)}(X^{(1)}_{i})\right\}\varepsilon^{(1)}_{i+n_{1}}\\
        &\quad+\left\{\mhat^{(2)}(X^{(1)}_{i})-m^{(2)}(X^{(1)}_{i})\right\}\left\{\mhat^{(2)}(X^{(1)}_{i+n_{1}})-m^{(2)}(X^{(1)}_{i+n_{1}})\right\}\Big]^{2}\left\{\sum_{d=1}^{p}k(X^{(1)}_{i}(d),X^{(1)}_{i+n_{1}}(d))\right\}^{2},\\
        Q_{2}&=\frac{2}{n_{1}}\sum_{i=1}^{n_{1}}\varepsilon^{(1)}_{i}\varepsilon^{(1)}_{i+n_{1}}\Big[-\varepsilon^{(1)}_{i}\left\{\mhat^{(2)}(X^{(1)}_{i+n_{1}})-m^{(2)}(X^{(1)}_{i+n_{1}})\right\}-\left\{\mhat^{(2)}(X^{(1)}_{i})-m^{(2)}(X^{(1)}_{i})\right\}\varepsilon^{(1)}_{i+n_{1}}\\
        &\quad+\left\{\mhat^{(2)}(X^{(1)}_{i})-m^{(2)}(X^{(1)}_{i})\right\}\left\{\mhat^{(2)}(X^{(1)}_{i+n_{1}})-m^{(2)}(X^{(1)}_{i+n_{1}})\right\}\Big]\left\{\sum_{d=1}^{p}k(X^{(1)}_{i}(d),X^{(1)}_{i+n_{1}}(d))\right\}^{2}. 
    \end{align*}
    By similar calculations as in the proof of Theorem \ref{thm:nulla}, one can derive
    \begin{align*}
        \E(|Q_{1}|\mid\Dsc^{(2)})&=o_{p}(n_{2}^{-1/2}p^{2}),\\
        \E(|Q_{2}|\mid\Dsc^{(2)})&=o_{p}(n_{2}^{-1/4}p^{2}).
    \end{align*}
    Using Lemma 6.1 in \citet{chernozhukov2018double}, 
    \begin{align*}
        Q_{1}&=o_{p}(n_{2}^{-1/2}p^{2}),\\
        Q_{2}&=o_{p}(n_{2}^{-1/4}p^{2}).
    \end{align*}
    As 
    \[
    \frac{1}{n_{1}}\sum_{i=1}^{n_{1}}\left\{\varepsilon^{(1)}_{i}\varepsilon^{(1)}_{i+n_{1}}\sum_{d=1}^{p}k(X^{(1)}_{i}(d),X^{(1)}_{i+n_{1}}(d))\right\}^{2}=\xi_{1}^{2}+O_{p}(n_{1}^{-1/2}p^{2}),
    \]
    and $\xi_{1}^{2}\asymp p^{2}$ under Assumptions \ref{assu1} and \ref{assu2p}, we have
    \[
    \frac{\xihat_{1}^{2}-\xi_{1}^{2}}{\xi_{1}^{2}}\convp 0.
    \]
    
    A similar analysis can be carried out for $\xihat_{2}^{2}$. The results follow.
\end{proof}

\begin{proof}[Proof of Theorem \ref{thm:alta}]
    Under $H_{a}$, we have
    \[
    \etahat^{(1)}_{i}=\eta^{(1)}_{i}+m^{(2)}(X^{(1)}_{i})-\mhat^{(2)}(X^{(1)}_{i}),
    \]
    and thus
    \[
    \sum_{d=1}^{p}\Deltahat^{(1)}_{d}=\frac{1}{n_{1}}\sum_{i=1}^{n_{1}}\eta^{(1)}_{i}\eta^{(1)}_{i+n_{1}}\sum_{d=1}^{p}k(X^{(1)}_{i}(d),X^{(1)}_{i+n_{1}}(d))+S_{1}+S_{2}+S_{3}+S_{4}+S_{5},
    \]
    where $S_{1},S_{2},S_{3}$ are defined as in the proof of Theorem \ref{thm:nulla}, and
    \begin{align*}
        S_{4}&=-\frac{1}{n_{1}}\sum_{i=1}^{n_{1}}\left\{m^{(1)}(X^{(1)}_{i})-m^{(2)}(X^{(1)}_{i})\right\}\left\{\mhat^{(2)}(X^{(1)}_{i+n_{1}})-m^{(2)}(X^{(1)}_{i+n_{1}})\right\}\sum_{d=1}^{p}k(X^{(1)}_{i}(d),X^{(1)}_{i+n_{1}}(d)),\\
        S_{5}&=-\frac{1}{n_{1}}\sum_{i=1}^{n_{1}}\left\{\mhat^{(2)}(X^{(1)}_{i})-m^{(2)}(X^{(1)}_{i})\right\}\left\{m^{(1)}(X^{(1)}_{i+n_{1}})-m^{(2)}(X^{(1)}_{i+n_{1}})\right\}\sum_{d=1}^{p}k(X^{(1)}_{i}(d),X^{(1)}_{i+n_{1}}(d)).
    \end{align*}
    Following the proof of Theorem \ref{thm:nulla}, we still have
    \begin{align*}
        S_{1}&=o_{p}(n_{1}^{-1/2}n_{2}^{-1/4}p),\\
        S_{2}&=o_{p}(n_{1}^{-1/2}n_{2}^{-1/4}p),\\
        S_{3}&=o_{p}(n_{2}^{-1/2}p).
    \end{align*}
    Moreover, under Assumptions \ref{assu1} and \ref{assu3}, 
    \begin{align*}
        \E(|S_{4}|\mid\Dsc^{(2)})&\le\E\left\{\left|m^{(1)}(X^{(1)})-m^{(2)}(X^{(1)})\right|\left|\mhat^{(2)}(X^{(1)\prime})-m^{(2)}(X^{(1)\prime})\right||G(X^{(1)},X^{(1)\prime})|\mid\Dsc^{(2)}\right\}\\
        &\le Kp\sqrt{\E\left\{\left|m^{(1)}(X^{(1)})-m^{(2)}(X^{(1)})\right|^{2}\right\}}\sqrt{\E\left\{\left|\mhat^{(2)}(X^{(1)\prime})-m^{(2)}(X^{(1)\prime})\right|^{2}\mid\Dsc^{(2)}\right\}}\\
        &=\|m^{(1)}-m^{(2)}\|_{L^{2}(P_{X}^{(1)})}o_{p}(n_{2}^{-1/4}p).
    \end{align*}
    Using Lemma 6.1 in \citet{chernozhukov2018double}, we have
    \[
    S_{4}=\|m^{(1)}-m^{(2)}\|_{L^{2}(P_{X}^{(1)})}o_{p}(n_{2}^{-1/4}p).
    \]
    Likewise, one can derive $S_{5}=\|m^{(1)}-m^{(2)}\|_{L^{2}(P_{X}^{(1)})}o_{p}(n_{2}^{-1/4}p)$. 
    
    A similar analysis can be carried out for $\sum_{d=1}^{p}\Deltahat^{(2)}_{d}$. Consequently, under $H_{a}$,
    \begin{align*}
        T_{a}&=\sum_{l=1}^{2}\frac{1}{n_{l}}\sum_{i=1}^{n_{l}}\eta^{(l)}_{i}\eta^{(l)}_{i+n_{l}}\sum_{d=1}^{p}k(X^{(1)}_{i}(d),X^{(1)}_{i+n_{1}}(d))+o_{p}(n_{1}^{-1/2}p+n_{2}^{-1/2}p)\\
        &\quad+\|m^{(1)}-m^{(2)}\|_{L^{2}(P_{X}^{(2)})}o_{p}(n_{1}^{-1/4}p)+\|m^{(1)}-m^{(2)}\|_{L^{2}(P_{X}^{(1)})}o_{p}(n_{2}^{-1/4}p)\\
        &=\sum_{d=1}^{p}\Delta^{(1)}_{d}+\sum_{d=1}^{p}\Delta^{(2)}_{d}+O_{p}(n_{1}^{-1/2}p+n_{2}^{-1/2}p)\\
        &\quad+\|m^{(1)}-m^{(2)}\|_{L^{2}(P_{X}^{(2)})}o_{p}(n_{1}^{-1/4}p)+\|m^{(1)}-m^{(2)}\|_{L^{2}(P_{X}^{(1)})}o_{p}(n_{2}^{-1/4}p).
    \end{align*}
    As $\xihat_{l}^{2}=\E[\{\eta^{(l)}\eta^{(l)\prime}G(X^{(l)},X^{(l)\prime})\}^{2}]\{1+o_{p}(1)\}$ for $l=1,2$, the results follow. 
\end{proof}

\end{appendices}

\newpage
\bibliography{main}

\begin{thebibliography}{}

\bibitem[Brooks et~al., 1989]{airfoil}
Brooks, T., Pope, D., and Marcolini, M. (1989).
\newblock {Airfoil Self-Noise}.
\newblock UCI Machine Learning Repository.
\newblock {DOI}: https://doi.org/10.24432/C5VW2C.

\bibitem[Chen and Guestrin, 2016]{chen2016xgboost}
Chen, T. and Guestrin, C. (2016).
\newblock Xgboost: A scalable tree boosting system.
\newblock In {\em Proceedings of the 22nd acm sigkdd international conference on knowledge discovery and data mining}, pages 785--794.

\bibitem[Chernozhukov et~al., 2018]{chernozhukov2018double}
Chernozhukov, V., Chetverikov, D., Demirer, M., Duflo, E., Hansen, C., Newey, W., and Robins, J. (2018).
\newblock Double/debiased machine learning for treatment and structural parameters.
\newblock {\em The Econometrics Journal}, 21(1).

\bibitem[Crump et~al., 2008]{crump2008nonparametric}
Crump, R.~K., Hotz, V.~J., Imbens, G.~W., and Mitnik, O.~A. (2008).
\newblock Nonparametric tests for treatment effect heterogeneity.
\newblock {\em The Review of Economics and Statistics}, 90(3):389--405.

\bibitem[Delgado, 1993]{delgado1993testing}
Delgado, M.~A. (1993).
\newblock Testing the equality of nonparametric regression curves.
\newblock {\em Statistics \& probability letters}, 17(3):199--204.

\bibitem[Dette and Neumeyer, 2001]{dette2001nonparametric}
Dette, H. and Neumeyer, N. (2001).
\newblock Nonparametric analysis of covariance.
\newblock {\em the Annals of Statistics}, 29(5):1361--1400.

\bibitem[Gonz{\'a}lez-Manteiga and Crujeiras, 2013]{gonzalez2013updated}
Gonz{\'a}lez-Manteiga, W. and Crujeiras, R.~M. (2013).
\newblock An updated review of goodness-of-fit tests for regression models.
\newblock {\em Test}, 22(3):361--411.

\bibitem[Gretton et~al., 2012]{gretton2012kernel}
Gretton, A., Borgwardt, K.~M., Rasch, M.~J., Sch{\"o}lkopf, B., and Smola, A. (2012).
\newblock A kernel two-sample test.
\newblock {\em The journal of machine learning research}, 13(1):723--773.

\bibitem[Hall and Hart, 1990]{hall1990bootstrap}
Hall, P. and Hart, J.~D. (1990).
\newblock Bootstrap test for difference between means in nonparametric regression.
\newblock {\em Journal of the American Statistical Association}, 85(412):1039--1049.

\bibitem[Hall and Heyde, 2014]{hall2014martingale}
Hall, P. and Heyde, C.~C. (2014).
\newblock {\em Martingale limit theory and its application}.
\newblock Academic press.

\bibitem[Han and Shen, 2024]{han2024generalized}
Han, Q. and Shen, Y. (2024).
\newblock Generalized kernel distance covariance in high dimensions: non-null clts and power universality.
\newblock {\em Information and Inference: A Journal of the IMA}, 13(3):iaae017.

\bibitem[He et~al., 2025]{he2025goodness}
He, C., Chen, C., and Zhu, L. (2025).
\newblock A goodness-of-fit assessment for general learning procedures in high dimensions.
\newblock {\em Journal of the American Statistical Association}, pages 1--12.

\bibitem[Hu and Lei, 2024]{hu2024two}
Hu, X. and Lei, J. (2024).
\newblock A two-sample conditional distribution test using conformal prediction and weighted rank sum.
\newblock {\em Journal of the American Statistical Association}, 119(546):1136--1154.

\bibitem[Joshi et~al., 2025]{joshi2025conformal}
Joshi, S., Kiyani, S., Pappas, G., Dobriban, E., and Hassani, H. (2025).
\newblock Conformal inference under high-dimensional covariate shifts via likelihood-ratio regularization.
\newblock {\em arXiv preprint arXiv:2502.13030}.

\bibitem[Kennedy, 2024]{kennedy2024semiparametric}
Kennedy, E.~H. (2024).
\newblock Semiparametric doubly robust targeted double machine learning: a review.
\newblock {\em Handbook of Statistical Methods for Precision Medicine}, pages 207--236.

\bibitem[King et~al., 1991]{king1991testing}
King, E., Hart, J.~D., and Wehrly, T.~E. (1991).
\newblock Testing the equality of two regression curves using linear smoothers.
\newblock {\em Statistics \& Probability Letters}, 12(3):239--247.

\bibitem[Kulasekera, 1995]{kulasekera1995comparison}
Kulasekera, K. (1995).
\newblock Comparison of regression curves using quasi-residuals.
\newblock {\em Journal of the American Statistical Association}, 90(431):1085--1093.

\bibitem[Lai et~al., 2021]{lai2021kernel}
Lai, T., Zhang, Z., and Wang, Y. (2021).
\newblock A kernel-based measure for conditional mean dependence.
\newblock {\em Computational Statistics \& Data Analysis}, 160:107246.

\bibitem[Lavergne, 2001]{lavergne2001equality}
Lavergne, P. (2001).
\newblock An equality test across nonparametric regressions.
\newblock {\em Journal of Econometrics}, 103(1-2):307--344.

\bibitem[Li et~al., 2023]{li2023testing}
Li, R., Xu, K., Zhou, Y., and Zhu, L. (2023).
\newblock Testing the effects of high-dimensional covariates via aggregating cumulative covariances.
\newblock {\em Journal of the American Statistical Association}, 118(543):2184--2194.

\bibitem[Munk and Dette, 1998]{munk1998nonparametric}
Munk, A. and Dette, H. (1998).
\newblock Nonparametric comparison of several regression functions: exact and asymptotic theory.
\newblock {\em Annals of statistics}, pages 2339--2368.

\bibitem[Neumeyer and Dette, 2003]{neumeyer2003nonparametric}
Neumeyer, N. and Dette, H. (2003).
\newblock Nonparametric comparison of regression curves: an empirical process approach.
\newblock {\em The Annals of Statistics}, 31(3):880--920.

\bibitem[Nie and Nicolae, 2022]{nie2022detection}
Nie, L. and Nicolae, D. (2022).
\newblock Detection and localization of changes in conditional distributions.
\newblock {\em Advances in Neural Information Processing Systems}, 35:36216--36229.

\bibitem[Pardo-Fern{\'a}ndez et~al., 2015a]{pardo2015tests}
Pardo-Fern{\'a}ndez, J.~C., Jim{\'e}nez-Gamero, M.~D., and El~Ghouch, A. (2015a).
\newblock Tests for the equality of conditional variance functions in nonparametric regression.
\newblock {\em Electronic Journal of Statistics}, 9:1826--1851.

\bibitem[Pardo-Fern{\'a}ndez et~al., 2015b]{pardo2015non}
Pardo-Fern{\'a}ndez, J.~C., Jim{\'e}nez-Gamero, M.~D., and Ghouch, A.~E. (2015b).
\newblock A non-parametric anova-type test for regression curves based on characteristic functions.
\newblock {\em Scandinavian Journal of Statistics}, 42(1):197--213.

\bibitem[Pardo-Fern{\'a}ndez et~al., 2007]{pardo2007testing}
Pardo-Fern{\'a}ndez, J.~C., Van~Keilegom, I., and Gonz{\'a}lez-Manteiga, W. (2007).
\newblock Testing for the equality of k regression curves.
\newblock {\em Statistica Sinica}, pages 1115--1137.

\bibitem[Racine and Van~Keilegom, 2020]{racine2020smooth}
Racine, J.~S. and Van~Keilegom, I. (2020).
\newblock A smooth nonparametric, multivariate, mixed-data location-scale test.
\newblock {\em Journal of Business \& Economic Statistics}, 38(4):784--795.

\bibitem[Ramdas et~al., 2015]{ramdas2015adaptivity}
Ramdas, A., Reddi, S.~J., Poczos, B., Singh, A., and Wasserman, L. (2015).
\newblock Adaptivity and computation-statistics tradeoffs for kernel and distance based high dimensional two sample testing.
\newblock {\em arXiv preprint arXiv:1508.00655}.

\bibitem[Redmond, 2002]{communities_and_crime}
Redmond, M. (2002).
\newblock {Communities and Crime}.
\newblock UCI Machine Learning Repository.
\newblock {DOI}: https://doi.org/10.24432/C53W3X.

\bibitem[Shao and Zhang, 2014]{shao2014martingale}
Shao, X. and Zhang, J. (2014).
\newblock Martingale difference correlation and its use in high-dimensional variable screening.
\newblock {\em Journal of the American Statistical Association}, 109(507):1302--1318.

\bibitem[Srihera and Stute, 2010]{srihera2010nonparametric}
Srihera, R. and Stute, W. (2010).
\newblock Nonparametric comparison of regression functions.
\newblock {\em Journal of Multivariate Analysis}, 101(9):2039--2059.

\bibitem[Sriperumbudur et~al., 2008]{sriperumbudur2008injective}
Sriperumbudur, B.~K., Gretton, A., Fukumizu, K., Lanckriet, G., and Sch{\"o}lkopf, B. (2008).
\newblock Injective hilbert space embeddings of probability measures.
\newblock In {\em 21st annual conference on learning theory (COLT 2008)}, pages 111--122. Omnipress.

\bibitem[Sriperumbudur et~al., 2010]{sriperumbudur2010hilbert}
Sriperumbudur, B.~K., Gretton, A., Fukumizu, K., Sch{\"o}lkopf, B., and Lanckriet, G.~R. (2010).
\newblock Hilbert space embeddings and metrics on probability measures.
\newblock {\em The Journal of Machine Learning Research}, 11:1517--1561.

\bibitem[Tibshirani et~al., 2019]{tibshirani2019conformal}
Tibshirani, R.~J., Foygel~Barber, R., Candes, E., and Ramdas, A. (2019).
\newblock Conformal prediction under covariate shift.
\newblock {\em Advances in neural information processing systems}, 32.

\bibitem[Wang et~al., 2025]{wang2025phase}
Wang, Y., Deb, N., and Mukherjee, D. (2025).
\newblock Phase transition in nonparametric minimax rates for covariate shifts on approximate manifolds.
\newblock {\em arXiv preprint arXiv:2507.00889}.

\bibitem[Yan et~al., 2022]{yan2022distance}
Yan, J., Li, Z., and Zhang, X. (2022).
\newblock Distance and kernel-based measures for global and local two-sample conditional distribution testing.
\newblock {\em arXiv preprint arXiv:2210.08149}.

\bibitem[Yan and Zhang, 2023]{yan2023kernel}
Yan, J. and Zhang, X. (2023).
\newblock Kernel two-sample tests in high dimensions: interplay between moment discrepancy and dimension-and-sample orders.
\newblock {\em Biometrika}, 110(2):411--430.

\bibitem[Young and Bowman, 1995]{young1995non}
Young, S.~G. and Bowman, A.~W. (1995).
\newblock Non-parametric analysis of covariance.
\newblock {\em Biometrics}, pages 920--931.

\bibitem[Zhang et~al., 2023]{zhang2023classification}
Zhang, J., Ding, J., and Yang, Y. (2023).
\newblock Is a classification procedure good enough?—a goodness-of-fit assessment tool for classification learning.
\newblock {\em Journal of the American Statistical Association}, 118(542):1115--1125.

\bibitem[Zhang et~al., 2018]{zhang2018conditional}
Zhang, X., Yao, S., and Shao, X. (2018).
\newblock Conditional mean and quantile dependence testing in high dimension.
\newblock {\em The Annals of Statistics}, 46(1):219--246.

\end{thebibliography}

\end{document}